\documentclass[conference]{IEEEtran}
\usepackage{graphicx}
\usepackage{amsthm}
\usepackage{epsfig}
\usepackage{latexsym}
\usepackage{amsfonts}
\usepackage{here}
\usepackage{rawfonts}
\usepackage[latin1]{inputenc}
\usepackage[T1]{fontenc}
\usepackage{calc}
\usepackage{capitalgreekitalic}
\usepackage{url}
\usepackage{enumerate}
\usepackage{color}
\usepackage[tbtags]{amsmath}
\usepackage{amssymb}
\usepackage{upref}
\usepackage{epic,eepic}
\usepackage{times}
\usepackage{dsfont}
\usepackage{comment}
\usepackage{cite}
\usepackage{bbm}

\renewcommand{\Pr}{\ensuremath{\operatorname{Pr}}}

\newtheorem{theorem}{\bf Theorem}

\newtheorem{definition}{\bf Definition}

\usepackage{dsfont}

\newcounter{step}
\newlength{\totlinewidth}
\newenvironment{algorithm}{%
  \rule{\linewidth}{1pt}
  \begin{list}{}%
    {\usecounter{step}%
      \settowidth{\labelwidth}{\textbf{Step 2:}}%
      \setlength{\leftmargin}{\labelwidth}%
      \setlength{\topsep}{-2pt}%
      \addtolength{\leftmargin}{\labelsep}%
      \addtolength{\leftmargin}{2mm}%
      \setlength{\rightmargin}{2mm}%
      \setlength{\totlinewidth}{\linewidth}%
      \addtolength{\totlinewidth}{\leftmargin}%
      \addtolength{\totlinewidth}{\rightmargin}%
      \setlength{\parsep}{0mm}%
      \raggedright}}%
  {\end{list}%
  \rule{\linewidth}{1pt}}
\newcounter{substep}

  {\end{list}}

\newlength{\aligntop}
\newlength{\alignbot}
 

\IEEEoverridecommandlockouts

\usepackage{algorithm}
\usepackage{algorithmic}

\usepackage{subfigure}

\begin{document}
\title{\huge Dynamic Psychological Game Theory for Secure Internet of Battlefield Things (IoBT) Systems}

\author{
\IEEEauthorblockN{Ye Hu,  Anibal Sanjab, and Walid Saad}  
\IEEEauthorblockA{\small Wireless@VT, Bradley Department of Electrical and Computer Engineering, Virginia Tech, Blacksburg, VA, USA. \\
Emails: yeh17@vt.edu, anibals@vt.edu and walids@vt.edu.}\vspace{-0.9cm}
\thanks{This research was sponsored in part by the U.S. National Science Foundation under Grant CNS-1446621 and, in part, by the Army Research Laboratory and was accomplished under Grant Number W911NF-17-1-0021. The views and conclusions contained in this document are those of the authors and should not be interpreted as representing the official policies, either expressed or implied, of the Army Research Laboratory or the U.S. Government. The U.S. Government is authorized to reproduce and distribute reprints for Government purposes notwithstanding any copyright notation herein.
}}

\maketitle

%

%

\vspace{0cm}
\begin{abstract}
In this paper, a novel anti-jamming mechanism is proposed to analyze and enhance the security of adversarial Internet of Battlefield Things (IoBT) systems. 
In particular, the problem is formulated as a dynamic psychological game between a soldier and an attacker. In this game, the soldier seeks to accomplish a time-critical mission by traversing a battlefield within a  certain amount of time, while maintaining its connectivity with an IoBT network. The attacker, on the other hand, seeks to find the optimal opportunity to compromise the IoBT network and maximize the delay of the soldier's IoBT transmission link. The soldier and the attacker's psychological behavior are captured using tools from psychological game theory, with which the soldier's and attacker's intentions to harm one another are considered in their utilities. 
To solve this game, a novel learning algorithm based on Bayesian updating is proposed to find a $\epsilon$-like psychological self-confirming equilibrium of the game. Simulation results show that, based on the error-free beliefs on the attacker's psychological strategies and beliefs, the soldier's material payoff can be improved by up to 15.11\% compared to a conventional dynamic game without psychological considerations.


\end{abstract}
%
%

\section{Introduction}

Emerging Internet of Things (IoT) technologies have led to significant changes in how autonomous systems are managed \cite{suri2016analyzing}. In a military environment, IoT technologies provide new ways for managing and operating a battlefield by interconnecting combat equipment, soldier devices, and other battlefield resources\cite{tortonesi2016leveraging}. This integration of the IoT with military networks is referred to as the \emph{Internet of Battlefield Things (IoBT)}\cite{suri2016analyzing}. In an IoBT, the connectivity between the wearables carried by the soldiers and other IoBT devices, such as multipurpose sensors, autonomous vehicles, and drones, plays a significant role in the mission-critical battlefield operations \cite{ray2015towards}.  However, the connectivity between these devices is highly vulnerable to cyber attacks, given the the adversarial nature of the battlefield coupled with the limitations of the IoBT devices' security mechanisms \cite{abuzainab2017dynamic}. Moreover, in an adversarial battlefield environment, the psychology of the soldiers and attackers could significantly influence their behavior, and, subsequently influence the security of the IoBT network. 


\subsection{Related Works}
The existing literature has studied a number of problems related to the security of the IoBT\cite{tortonesi2016leveraging,ray2015towards,abuzainab2017dynamic,Nof2018misinfo}. In \cite{tortonesi2016leveraging}, the communications and information management challenges of the IoBT are investigated. The work in\cite{ray2015towards} integrates IoT and network centric warfare for the enhancement of the IoBT integrity. The authors in \cite{abuzainab2017dynamic} use a feedback Stackelberg solution to dynamically optimize the connectivity of an adversarial IoBT network. The work in \cite{Nof2018misinfo} develops a mean-field game approach to analyze the spread of misinformation in an adversarial IoBT. Despite the promising results, these existing works\cite{tortonesi2016leveraging,ray2015towards,abuzainab2017dynamic} mostly rely on static constructs and do not consider the influence of the human players' psychology and potential bounded rationality when making decisions or choosing strategies within an IoBT setting. Indeed, the behavioral aspect of human decision making processes, leading agents to deviate from the fully rational objective behavior in an IoBT, has a direct impact on the security of the IoBT network. Hence, this aspect must be accounted for and thoroughly studied within the context of studying and assessing the security of the IoBT.


Recently, there has been significant interest in studying human behavior and its cyber-psychical security impact. The authors in \cite{hota2016fragility} study a common-pool resource game that captures the players' risk preference using tools from prospect theory.  The work in \cite{sanjab2017prospect} uses prospect theory to study the effect of a defender's and attacker's subjective behavior on the security of a drone delivery system. The work in \cite{xiao2017cloud} uses prospect theory to analyze the interaction between the defender of a cloud storage system and an attacker targeting the system with advanced persistent threats. In \cite{sanjab2016bounded}, a cognitive hierarchy theory based approach is proposed to capture the bounded rationality of defenders and attackers in  cyber-physical systems. These previous works present interesting and novel results. However, the existing literature has not yet considered and analyzed the influence of players' psychology on the game-theoretic decision making in IoT networks.  In fact, recent works in the game theory literature have shown that decision making is strongly impacted by human psychology  and have studied various games' aspects and solutions while accounting for psychological factors\cite{geanakoplos1989psychological,battigalli2009dynamic, battigalli2015frustration,rossi2017much}. In this regard, the work in \cite{geanakoplos1989psychological} proves the existence of sub-game perfect and sequential equilibria in psychological games. The authors in \cite{battigalli2009dynamic} study a game-theoretic model that captures dynamic psychological effects and develops new psychological game solution concepts. The work in \cite{battigalli2015frustration} considers the behavioral consequences of psychology in presence of blaming behaviors. In addition, the effect of the human psychology in mean-field-type games is studied in \cite{rossi2017much}. Despite the promising results, these existing works on psychological game theory and its applications\cite{geanakoplos1989psychological, battigalli2009dynamic, battigalli2015frustration,rossi2017much} have not analyzed the potential adoption of psychological game approaches in security scenarios. In \cite {ye2018psychology}, we studied how a soldier's and an attacker's psychology can impact an IoBT network's security. However, in \cite {ye2018psychology}, the players' resource limitations and IoBT connectivity objectives are not considered. In addition, in \cite {ye2018psychology}, the soldier's actions at each step in the battlefield reveal the soldier's preference on its future actions, as such the psychological forward induction of \cite {ye2018psychology} can be used to solve the proposed security problem.  Yet, in a real battlefield, the soldier' actions can be rather independent at each time step, making the psychological forward induction based solution of \cite {ye2018psychology} infeasible. 
Thus, there is a need to introduce new solutions that dynamically predict and react to the actions of adversaries in the battlefield, while taking the players' resource limitation and IoBT connectivity objectives into consideration.

\subsection{Contributions}
The main contribution of this paper is to analyze the psychological behavior of human decision makers in an adversarial IoBT network, in presence of stringent resource limitations (i.e. time limitation and power limitation) for the players. To our best knowledge, \emph{this is the first work that jointly considers players' resource limitation and their psychological behavior for securing an IoBT network}. Our key contributions include: 

\begin{itemize}

\item We develop a novel framework to dynamically optimize the connectivity between a soldier and the IoBT network. 
We consider a battlefield in which a soldier must accomplish a time-critical mission that requires traversing the battlefield while maintaining connectivity with the IoBT network. Meanwhile, the attacker in the battlefield is interested in compromising the soldier's IoBT connectivity, by selectively jamming the IoBT network at each time instant in the battlefield. The solider, acting as a defender, will selectively connect to certain IoBT devices at each time instant along its mission path, so as to evade the attack.

\item  We formulate this IoBT security problem as a dynamic game, in which the soldier attempts to predict and evade the attacker's attack at each time instant in the battlefield to minimize its cumulative expected retransmission delay, while the attacker aims at optimally targeting the IoBT devices to maximize the soldier's retransmission delay while accounting for its limited cumulative jamming power. Both the soldier's time limitation and the attacker's power limitation are considered in the formulated game. In this regard, we prove the uniqueness of the Nash equilibrium (NE) of this game, under a set of defined conditions, and we study the resulting NE strategies, which allows analysis of the optimal decision making processes of the soldier and attacker based on their built set of beliefs over the strategy on their opponent's strategies.



\item We perform fundamental analysis on the soldier's and attacker's psychology in the battlefield using the framework of psychological game theory\cite{battigalli2009dynamic}. In the formulated psychological game, the psychology of the players (i.e. the soldier and the attacker) is modeled as their intention to \emph{frustrate} each other. The frustration of the players is quantified as the gap, if positive, between their expected payoff and actual payoff. A psychological equilibrium (PE) is used to solve the psychological IoBT game. In this regard, we prove the uniqueness of the PE for our proposed psychological game, under the same set of conditions at which the NE is unique. In addition, our analytical results show that, in an attempt to frustrate the soldier, at the PE, the attacker is more prone to attack the IoBT device with the best channel conditions.

\item  We propose a Bayesian updating algorithm to establish the players' belief system, so as to solve the proposed psychological IoBT game. In this regard, the algorithm characterizes what is known as an $\epsilon$-like psychological self-confirming equilibrium (PSCE) of our proposed psychological game.

\item The results also show that, based on its error-free beliefs on the attacker's psychological strategies and beliefs, the soldier can obtain an up to $15.11\%$ gains in its expected material payoff at equilibrium, compared to a conventional dynamic game. Meanwhile, using Bayesian updating, the soldier and the attacker can achieve $\epsilon$-like beliefs, such that an $\epsilon$-like self-confirming psychological equilibrium of the formulated psychological game can be reached. Simulation results also show that, the non-error-free beliefs, which result from, for example, $10$ iterations in the Bayesian updating algorithm,  can yield up to $9.23\%$ loss in terms of the soldier's expected material payoff. 

\end{itemize}

The rest of this paper is organized as follows. The system model and problem formulation are described in Section \uppercase\expandafter{\romannumeral2}. The psychological analysis of the soldier and attacker is represented in Section \uppercase\expandafter{\romannumeral3}. The Bayesian updating-based solution of the psychological IoBT game is proposed in Section \uppercase\expandafter{\romannumeral4}. In Section \uppercase\expandafter{\romannumeral5}, simulation results are presented and analyzed. Finally, conclusions are drawn in Section \uppercase\expandafter{\romannumeral6}.
 
\section{System Model and Problem Formulation}\label{se:system}

Consider a battlefield in which a soldier seeks to move from an origin $O$ to a destination $D$ along a predefined path, using a minimum amount of time as shown in Fig. \ref{figure1}. At the same time, this soldier tries to communicate with a total of $X$ IoBT devices in a set $\mathcal{X}$ that is uniformly deployed along this path, to get access to situational awareness within the battlefield and to receive instructions from the battlefield commander. The soldier can only associate with one IoBT device at each location. The soldier should communicate with $J<X$ IoBT devices along the path, so as to maintain its total downlink transmission delay lower than $\Delta$, while getting access to the required information on time.
Meanwhile, in this battlefield, an attacker seeks to disrupt the connectivity between the soldier and the IoBT devices by jamming the communication links. Given the limitation on its total power $E$, the attacker can only compromise (i.e. jam) the IoBT network at most $J'$ times along the path. At each step in this battlefield, the soldier and attacker will sequentially choose strategies to realize their objective, based on their perfect observation on what happened in the battlefield. Here, the soldier's objective is minimizing its communication delay, the attacker's objective is maximizing the soldier's communication delay, while minimizing its total power consumption.
\begin{figure}[!t]
  \begin{center}
   \vspace{0cm}
    \includegraphics[width=8cm]{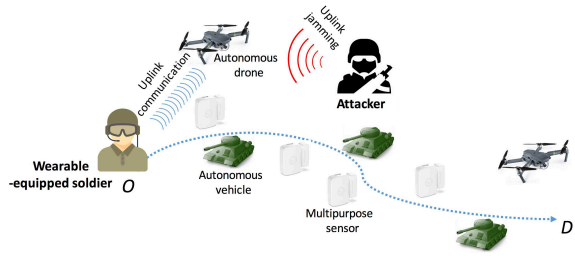}
    \caption{\label{figure1} Soldier battlefield security graph.}
  \end{center}
\end{figure}
\subsection{Soldier's communication delay}
We assume that the soldier (attacker) chooses to connect with (jam) the IoBT network at each step, sequentially, until the soldier arrives at $D$.
The soldier communicates with each IoBT device $x\in\mathcal{X}$ over a downlink channel $c_x$. The signal-to-interference-plus-noise ratio (SINR) $\gamma_{x}$ of the downlink channel between the soldier and IoBT device $x$ is given by:
    \begin{equation}\label{eq:1}
\gamma_{x} = \left\{ {\begin{array}{*{20}{c}}
{\frac{{{P_S}l^s_{x}}}{{{P_A}{l^a_{x}} + {\sigma ^2}}},~\textrm{ if~channel}~c_x~\textrm{is~jammed},}\\
{\frac{{{P_S}l^s_{x}}}{{{\sigma ^2}}},~  \textrm{otherwise},}
\end{array}} \right.
 \end{equation}
where $P_S$  and $P_A$ are, respectively, the transmit powers of soldier and the attacker. $l^s_{x}={g_{s}}{d^{-\lambda }_{s}}$ is the path loss between the soldier and IoBT device $x$, with $g_{s}$ being the Rayleigh fading parameter, $d_{s}$ being the distance between the soldier and IoBT device $x$, and $\lambda$ the path loss exponent. $l^a_{x}={g_{a}}{d^{-\lambda }_{a}}$ is the path loss between the attacker and the IoBT device $x$, with $g_{a}$ being the Rayleigh fading parameter, and $d_{a}$ being the distance between the attacker and the IoBT device $x$. $\sigma^2$ is the power of the Gaussian noise. 
At each step $x$, based on the probability distribution of the Rayleigh fading parameter $g_{s}$ and $g_{a}$, the probability that the soldier's received SINR, $\gamma_x$,  is higher than a threshold, $\hat{\gamma}$, in one time slot is given by: 
    \begin{equation}
{q_x} = \left\{ {\begin{array}{*{20}{c}}
{\int_V^\infty  {f\left( {{g_s}} \right)d{g_s}} },~\textrm{ if~channel}~c_x~\textrm{is~not~jammed},\\
{\int_0^\infty  {\int_0^{\frac{{{g_s} - V}}{W}} {f\left( {g{_a}} \right)f\left( {{g_s}} \right)dg{_a}d{g_s}} } },~\textrm{otherwise},
\end{array}} \right.
 \end{equation}
where ${V} = {\frac{{\widehat \gamma {\sigma ^2}}}{{{d_s}^{ - \lambda }{P_S}}}}$, $W = \frac{{{d^{-\lambda }_{a}}{P_A}\widehat \gamma }}{{d_s^{ - \lambda }{P_S}}}$. Here, $f\left( {{g_s}} \right)$ and $f\left( {{g_a}} \right)$ are the probability density functions of the Rayleigh fading parameters $g_s$ and $g_a$, respectively.  
In the studied battlefield, the soldier attempts to maintain a probability of achieving an SINR exceeding $\hat{\gamma}$, $q_x$, that is higher than a threshold $\hat{q}$.  Hence, the soldier will request $k$ retransmissions of the downlink data from IoBT device $x$. 
However, the soldier will perform $k<\widehat k$ retransmissions, in the case that the channel is occasionally experiencing a small scale fading. 
Thus, $k$ is given by:
 \begin{equation}
k =\left\{ {\begin{array}{*{20}{c}}
{\left \lceil \frac{{\log \left( {1 - \widehat q} \right)}}{{\log \left( {1 - {q_x}} \right)}} \right \rceil},~\textrm{ if}~k<\widehat k,\\
\widehat k, ~\textrm{ if}~k \ge \widehat k,
\end{array}} \right.
 \end{equation}
under the goal of maintaining $\left(1-q_x\right)^k<1-\widehat q$. Thus, the retransmission delay $\tau$ of the soldier at each step is given by $kt_x$. Here, $t_x=\frac{S}{{I_x\log \left( {1 + {\gamma _x}} \right)}}$ is the average unit transmission delay, which is the average duration of a successful packet transmission at the physical medium of one resource block, at step $x$ \cite{yang2008achieving}. $S$ is the size of one resource block, $I_x$ is the bandwidth of channel $c_x$. 
 
\subsection{Strategies of the players}
In the studied battlefield, the objective of the soldier is to effectively maintain a low transmission delay. Thus, the soldier will attempt to communicate with the IoBT devices that will not be attacked. $\mathcal{A}=\left\{ {{a}_1},{{a}_2} \right\}$ represents the soldier's action space at each step $x \in \mathcal{X}$. Here, at each step $x$, ${{a}_1 \in {\mathcal{A}}}$ indicates that the soldier builds a communication link with IoBT device $x$, whereas ${{a}_2 \in {\mathcal{A}}}$ indicates that the soldier does not communicate with IoBT device $x$. 

On the other hand, the objective of the attacker is to increase the retransmission delay of the soldier. As such, under a constraint on its total power consumption, the attacker will find the best time instant to launch an attack on the IoBT network, so as to decrease the SINR of the communication channel between the soldier and IoBT network. 
 The attacker's set of the possible actions at each step $x$ can be represented by $\mathcal{B}=\left\{ {{b}_1},{b}_2 \right\}$.  Here, action ${{b}_1} \in \mathcal{B}$ indicates that the attacker chooses to compromise the IoBT network, while action ${{b}_2} \in \mathcal{B}$ indicates that the attacker does not jam the IoBT network. Note that the jamming power, $P_K$, that will be used by the attacker is assumed to be constant. 
  
  In addition, we use ${h}^x$ to represent the sequence of actions that have been taken by each of the players before reaching step $x$. We refer to $h^x$ as the history at step $x$. In this respect, the set of all possible histories ${h}^x$ at step $x$ is denoted by $\mathcal{H}^x$. In addition, we let $\hat{h}^x$ denote the sequence of actions that have been taken by each player up to step $x$, including the action pair taken at step $x$.  After observing history ${h}^x \in \mathcal{H}^x$ at step $x$, the soldier and attacker will find the optimal strategies at the current step to realize their individual objectives. The set of soldier's feasible actions at history $h^x$ is, then, represented by $\mathcal{A}_{h^x}$, while the set of the attacker's feasible actions at history $h^x$ is represented by $\mathcal{B}_{h^x}$. The actions that are chosen by the soldier and the attacker at history $h^x$ are represented, respectively, by $a\left(h^x\right)$ and $b\left(h^x\right)$. In addition, the set of possible terminal histories $\hat{h}^X$, at which point the soldier reaches $D$, is represented by $\mathcal{Z}$, where $\mathcal{Z}=\left\{\mathcal{H}^X,(\mathcal{A}_{h^X},\mathcal{B}_{h^X})\right\}$.
  

In an adversarial IoBT environment, the soldier will randomize its action selection at each history such as to make it more complex for the attacker to guess the IoBT device to which the soldier aims to connect. The soldier will, hence, choose a probability distribution $\boldsymbol\alpha_{{h}^x}=\left[\alpha_1,\alpha_2\right]$ over its feasible action set $\mathcal{A}_{h^x}$ at history ${h}^x$. In this regard, $\alpha_i$ denotes the probability of choosing action $a_i\in\mathcal{A}_{h^x}$ at history $h^x$, where $i\in\{1,2\}$.  This probability distribution $\boldsymbol\alpha_{{h}^x}$ denotes the soldier's \emph{mixed strategy} at history ${h}^x$. A possible strategy for the soldier in the battlefield can, then, be represented by a set ${\alpha} = \left\{\boldsymbol{\alpha}_{{h}^x}\left| {{h}^x\in\mathcal{H}}^x, x\in \mathcal{X} \right. \right\}$. The set of all feasible strategies of the soldier is denoted by $\mathcal{C}$.
 
A similar randomization logic is used by the attacker. The attacker seeks to choose a probability distribution $\boldsymbol\beta_{{h}^x}=\left[\beta_1,\beta_2\right]$ over its feasible action set $\mathcal{B}_{h^x}$ at each history ${h}^x$, so as to maximize the soldier's transmission delay while keeping its consumed jamming power at a minimum. In this respect, $\beta_i$ corresponds to the probability of choosing action $b_i\in\mathcal{B}_{h^x}$ at history $h^x$, where $i\in\{1,2\}$. This probability distribution $\boldsymbol\beta_{{h}^x}$ is the attacker's mixed strategy at history ${h}^x$. A possible strategy for the attacker can, hence, be denoted by a set ${\beta} = \left\{\boldsymbol{\beta}_{{h}^x}\left| {{h}^x\in\mathcal{H}^x}, x\in \mathcal{X} \right. \right\}$. The set of all possible strategies of the attacker is denoted by $\mathcal{D}$.

\subsection{Material payoff}
We define the soldier's \emph{material payoff} as the normalized gap between the sum of the soldier's actual communication delay and the soldier's maximum tolerable communication delay. Meanwhile, we define the attacker's \emph{material payoff} as the weighted sum of the soldier's time delay and the attacker's power consumption. 

Note that, in (2), $q_x$ depends on both the soldier and attacker's actions $a\left(h^x\right)$ and $b\left(h^x\right)$ in the form: 
  \begin{equation}
\begin{split}
q_x&\left(a\left(h^x\right)=a_1, b\left(h^x\right)\right)=\left(1-{\mathbbm{1}_{b\left(h^x\right)=b_1}}\right){\int_V^\infty  {f\left( {{g_x}} \right)d{g_x}} }\\
&+{\mathbbm{1}_{b\left(h^x\right)=b_1}}{\int_0^\infty  {\int_0^{\frac{{{g_x} - V}}{W}} {f\left( {g{_a}} \right)f\left( {{g_x}} \right)dg{_a}d{g_x}} } },
\end{split}
 \end{equation}
 where ${\mathbbm{1}_{b\left(h^x\right)=b_1}}$ is an indicator function that only equals to $1$ when the current action of the attacker is $b_1$.
 Hence, the soldier's time delay, when attempting to communicate with $x\in\mathcal{X}$ at history $h^x$, is given by:
 \begin{equation}
 \begin{split}
 \tau& \left( {a\left( h^x \right)=a_1,b\left( h^x \right)} \right) \\
 &= \left\{ {\begin{array}{*{20}{c}}
{ \left\lceil {\frac{{\log \left( {1 - \hat q} \right)}}{{\log \left( {1 - {q_x}\left( {a\left( h^x \right),b\left( h^x \right)} \right)} \right)}}} \right\rceil t,~\textrm{if}~k < \hat k,}\\
{\hat kt,~\textrm{otherwise}.}
\end{array}} \right.
 \end{split}
 \end{equation}

 In case the soldier does not communicate with the IoBT device at history $h^x$, the soldier will naturally not incur any delay which leads to $\tau(a(h^x)=2,b(h^x))=0$.
%
%
%
As such, the soldier's retransmission delay is a function of the soldier's and attacker's actions. At the terminal history $\hat{h}^X$, the soldier's accumulated communication delay will be:
\begin{equation}
 \tau\left({\hat{h}}^X\right) =\sum\limits_{v = 1}^X {\tau {\left({a}\left(h^v\right),{b\left(h^v\right)}\right)} },
 \end{equation}
   where ${a}\left(h^v\right)$ and ${b}\left(h^v\right)$ represent, respectively, the soldier's and attacker's action at step $v$ in ${\hat{h}}^X$. Note that, under each terminal history $\hat{h}^X \in \mathcal{Z}$, $\sum\limits_{v = 1}^{X}{\mathbbm{1}_{a\left(h^v\right)=a_1}}=J$ and $\sum\limits_{v = 1}^{X}{P_K\mathbbm{1}_{b\left(h^v\right)=b_1}} \le E$. Here, we note that, even though not communicating with any device will lead to a minimum delay for the soldier, this is not a feasible strategy for the soldier, since by definition, the soldier has to communicate with $0<J\leq X$ devices in the battlefield so as to acquire situational awareness. Based on its primary objective, the soldier will determine an optimal strategy ${\alpha}$ that minimize its expected total time delay. This, hence, requires maximizing the normalized gap between the cumulative retransmission delay and the maximum tolerable delay, which is defined as :
  \begin{equation}
\begin{split}
&{{\pi}} \left( {{{\alpha}, {\beta}}} \right)=\frac{\Delta-\sum\limits_{{\hat{h}^X} \in \mathcal{Z}} Q_{{\alpha}, {\beta}} \left({\hat{h}}^X\right) {\tau {\left({\hat{h}}^X\right)}  }}{\Delta},
\end{split}
 \end{equation}
 where $Q_{{\alpha}, {\beta}}\left({\hat{h}}^X\right)$ is the probability of occurrence of terminal history $\hat{h}^X \in \mathcal{Z}$, and is induced by the soldier's and the attacker's mixed-strategies, $\alpha$ and $\beta$ as follow:
   \begin{equation}
\begin{split}
{Q_{\alpha ,\beta }}\left( {{{\hat{h}}^X}} \right) = &\mathop \prod \limits_{x = 1}^{X} \left( {{\alpha _{{h^x}}}\left( 1 \right)  {\mathbbm{1}_{a\left( h^x \right) = {a_1}}} + {\alpha _{{h^x}}}\left( 2 \right) {\mathbbm{1}_{a\left( h^x \right) = {a_2}}}} \right)  \\
&\times\left( {{\beta _{{h^x}}}\left( 1 \right) {\mathbbm{1}_{b\left( h^x \right) = {b_1}}} + {\beta _{{h^x}}}\left( 2 \right) {\mathbbm{1}_{b\left( h^x \right) = {b_2}}}} \right),
\end{split}
 \end{equation}
where history ${h}^x$ is part of $\hat{h}^X$. In other words, $h^x$ represents the sequence of actions in $\hat{h}^X$ taken before $x$. Hence, $\pi(\alpha,\beta)$ represents the soldier's expected utility (or, equivalently, expected material payoff) achieved under the strategy pair $(\alpha,\beta)$.
 
Meanwhile, the material payoff of the attacker at the terminal history ${\hat{h}}^X$ is defined as 
\begin{equation}
\begin{split}
\pi'_0&\left({\hat{h}}^X\right) =\\
&\theta_1\frac{{\sum\limits_{v = 1}^X {\tau {\left({a}\left(h^v\right),{b\left(h^v\right)}\right)} } }}{\Delta }+\theta_2  \frac{{E - \sum\limits_{v = 1}^X{P_K}  {\mathbbm{1}_{{b}\left(h^v\right) = {b_1}}}}}{E},
\end{split}
\end{equation}
where $\theta_1 \ge 0$ and $\theta_2 \ge 0$ represent, respectively, the weight of time delay and power consumption, with $\theta_1+\theta_2=1$. Thus, in this battlefield, the attacker will select the optimal strategy ${\beta}$ that maximizes the soldier's time delay\footnote{Here, we assume that $\Delta$ can be learnt by the attacker using its knowledge of the IoBT devices' quantity and channel condition, or through, for example, a prior reconnaissance phase about the soldier and its objectives.}, while minimizing its power consumption, which can be captured by maximizing the following expected utility (i.e. expected material payoff):
  \begin{equation}
\begin{split}
&\pi'\left( {{{\alpha}, {\beta}}} \right) =\sum\limits_{{\hat{h}^X} \in \mathcal{Z}} Q_{{\alpha},{\beta}}\left({\hat{h}}^X\right)  {\pi'_0\left({\hat{h}}^X\right)  }.
\end{split}
 \end{equation}



In this battlefield, 
the attacker can track the soldier's location via GPS, and it can gather intelligence (i.e., knowledge) on the soldier's associated objective. However, it does not know the IoBT devices to which the soldier will connect. Meanwhile, the soldier knows that the attacker is present, but does not know which IoBT devices it will target. 
Then, to determine their optimal actions at each history, the soldier and attacker aim at forming an estimation of their opponent's actions (e.g. the attacker estimates the soldier's actions, and the soldier estimates the attacker's actions). This estimation is defined as the soldier and attacker's \emph{first-order beliefs} on each other. 
Let $\boldsymbol{\delta}^1_{h^x}  = \left[ {{\delta ^1_{h^x}}\left( 1 \right),{\delta ^1_{h^x}}\left( 2 \right)} \right]$ be the soldier's vector of beliefs on the probability distribution of the attacker's actions  ${{b}_1}$ and $b_2$ at history $h^x$, and let $\boldsymbol{\rho} ^1_{h^x} = \left[ {\rho^1_{h^x}\left( 1 \right),\rho^1_{h^x}\left( 2 \right)} \right]$ be the attacker's belief vector on the probability distribution of the soldier's actions ${{a}_1}$ and $a_2$ at history $h^x$, respectively.  As such, we let $\delta^1_{\mathcal{H}^x}$ and $\rho^1_{\mathcal{H}^x}$, denote the set of first-order beliefs of, respectively, the soldier and attacker for each possible history at step $x$.  Hereinafter, we use ${\delta}^1=\left\{\boldsymbol{\delta}  ^1_{\mathcal{H}^1}, \cdots, \boldsymbol{\delta}  ^1_{\mathcal{H}^X}\right\}$ to denote a set of soldier's first-order beliefs on the attacker, and ${\rho}^1=\left\{\boldsymbol{\rho}  ^1_{\mathcal{H}^1}, \cdots, \boldsymbol{\rho}  ^1_{\mathcal{H}^X}\right\}$ to denote a set of attacker's first-order beliefs on the soldier, at each possible history.

 Based on belief ${\delta}^1$, the soldier's perceived (i.e. belief-based) expected material payoff will be given by:
  \begin{equation}
\begin{split}
{\overline{\pi}} \left( {{{\alpha},{\delta}^1}} \right)=\frac{\Delta-\sum\limits_{{\hat{h}^X} \in \mathcal{Z}} Q_{{\alpha}, {\delta}^1}{\left({\hat{h}}^X\right)} {\tau {\left({\hat{h}}^X\right)}  }}{\Delta},
\end{split}
 \end{equation}
 where $Q_{{\alpha}, {\delta}^1}{\left({\hat{h}}^X\right)}$ is the belief-based probability of occurrence of the terminal history $\hat{h}^X \in \mathcal{Z}$ induced by ${\alpha}$ and ${\delta}^1$:
    \begin{equation}
\begin{split}
{Q_{\alpha ,{\delta}^1 }}\left( {{{\hat{h}}^X}} \right) = &\mathop \prod \limits_{x = 1}^{X} \left( {{\alpha _{{h^x}}}\left( 1 \right)  {\mathbbm{1}_{a\left( h^x \right) = {a_1}}} + {\alpha _{{h^x}}}\left( 2 \right) {\mathbbm{1}_{a\left( h^x \right) = {a_2}}}} \right)\\
&\times\left( {{{\delta}^1 _{{h^x}}}\left( 1 \right) {\mathbbm{1}_{b\left( h^x \right) = {b_1}}} + {{\delta}^1 _{{h^x}}}\left( 2 \right) {\mathbbm{1}_{b\left( h^x \right) = {b_2}}}} \right).
\end{split}
 \end{equation}
 
 

 Similarly, given its first-order belief ${\rho}  ^1$, the attacker's perceived (i.e. belief-based) expected material payoff under strategy ${\beta}$ will be:
    \begin{equation}\small
\begin{split}
\overline\pi'\left( {{{\beta}, {\rho}  ^1}} \right) =\sum\limits_{{\hat{h}^X} \in \mathcal{Z}} Q_{{\rho}  ^1, {\beta} } {\left({\hat{h}}^X\right)} {\pi'_0\left({\hat{h}}^X\right)  },
\end{split}
 \end{equation}
  where $Q_{{\rho}  ^1, {\beta} }{\left({\hat{h}}^X\right)} $ is the belief-based probability of occurrence of the terminal history $\hat{h}^X \in \mathcal{Z}$ induced by ${\rho}  ^1$ and ${\beta}$:
  
     \begin{equation}
\begin{split}
{Q_{{\rho}  ^1 ,\beta }}\left( {{{\hat{h}}^X}} \right) = &\mathop \prod \limits_{x = 1}^{X} \left( {{{\rho}  ^1 _{{h^x}}}\left( 1 \right)  {\mathbbm{1}_{a\left( h^x \right) = {a_1}}} + {{\rho}  ^1 _{{h^x}}}\left( 2 \right) {\mathbbm{1}_{a\left( h^x \right) = {a_2}}}} \right)\\
&\times \left( {{\beta _{{h^x}}}\left( 1 \right) {\mathbbm{1}_{b\left( h^x \right) = {b_1}}} + {\beta _{{h^x}}}\left( 2 \right) {\mathbbm{1}_{b\left( h^x \right) = {b_2}}}} \right).
\end{split}
 \end{equation}  
  \subsection{Game formulation}
In the studied IoBT scenario, the primary objective of the soldier is to find a strategy to effectively evade the jamming attack of the attacker, while the objective of the attacker is to find an attack strategy that effectively jams the soldier's communication with the IoBT devices. 
As such, we formulate a dynamic game $[\mathcal{P},\mathcal{H}, \mathcal{Z}, \pi, \pi', \overline\pi, \overline\pi']$ to capture the dependence between the objectives and the actions of the soldier and the attacker. Here, $\mathcal{P}$ is the set of players which includes the soldier and attacker. 
$\mathcal{H}$ is the set of histories representing the sequence of actions that have been taken by each of the players before reaching a certain stage, 
and $\mathcal{Z}$ represents the set of terminal histories, at which point the soldier reaches its destination, $D$, and the game ends. $\pi$ and $\pi'$ are the expected utilities of the soldier and the attacker, respectively, defined in (7) and (10), while $\bar{\pi}$ and $\bar{\pi}'$ are their belief-dependent (i.e. perceived) expected utilities, defined in (11) and (13). In this formulated game, each of the soldier and the attacker aim at maximizing their (belief-based) expected utilities. When the beliefs of each player accurately predict the strategy of the opponent, and when each player chooses a strategy that maximizes its expected utility based on those beliefs, these strategies give rise to a \emph{Nash equilibrium (NE)} for the proposed game, which is formally defined as follows:

\begin{definition}\emph{A \emph{Nash equilibrium (NE)} for the formulated dynamic game is defined as $\left({\alpha}^*, {\beta}^*, {\rho}^{1*}, {\delta}^{1*}\right)$,  in which ${\alpha}^*$, and ${\beta}^*$ are rational, such that: 
\begin{equation}
{\alpha}^* \in \mathop {\arg \max }\limits_{{\alpha}\in \mathcal{C}} {\overline{\pi}} \left( {{{\alpha},{\delta}^{1*}}} \right),
\end{equation}
\begin{equation}
{\beta}^* \in \mathop {\arg \max }\limits_{{\beta} \in\mathcal{D}} \overline\pi'\left( {{{\beta}, {\rho}  ^{1*}}} \right),
\end{equation}
while beliefs ${\rho}^{1*}$ and ${\delta}^{1*}$ are error-free such that 
for all ${a}_n \in \mathcal{A}_{h^x}$ at each history ${h}^x$ in the game: 
\begin{equation}
{{\rho}_{{h}^x} ^{1*}}\left(n \right) = {\alpha}^{*}_{{h}^x}\left(n \right),
\end{equation} 
and for all $b_m \in \mathcal{B}_{h^x}$ at each history ${h}^x$ in the game:
\begin{equation}
{{\delta}^{1*}_{{h}^x}}\left( m \right) =  {\beta}^*_{{h}^x}\left( m \right).
\end{equation}}
\end{definition}
 Thus, at an NE of the proposed game, both the soldier and attacker correctly estimate their opponents' strategies (represented by an error-free set of beliefs over the opponent's strategy) and make rational determinations on their strategies based on their error-free beliefs, at every history of the game. As shown in Definition 1, the rational strategies of the players (i.e. the soldier and the attacker) are the strategies that maximize the players' error-free belief-based perceived expected payoff, $\overline\pi$ and $\overline\pi'$. At each history $h^x\in\mathcal{H}^x$, the players' error-free first-order beliefs (i.e. $\delta^{1*}$ and $\rho^{1*}$) on each of their opponents' feasible action (i.e. $b_m\in\mathcal{B}_{h^x}$, $a_n\in\mathcal{A}_{h^x}$) equals the probability that their opponents choose this action with their rational strategies (i.e. $\alpha^*$ and $\beta^*$). The players, including the soldier and the attacker, are considered to hold accurate (error-free) beliefs in the computation of their respective NE strategies. Hence, these error-free beliefs require that, at equilibrium, beliefs should accurately predict the opponent's strategy. However, in practical networks, the players' beliefs may not be fully accurate, when no effective prediction method is used. Hence, when solving (15) and (16), the resulting soldier and attacker strategies are rational (i.e. optimal), but are based on their respective beliefs. If these beliefs are not accurate (i.e. if (17) and (18) are not met), even through each of the players is still acting rationally, their strategies may deviate from the NE strategies.   

 Moreover, in practice, as emotional human players, the soldier and the attacker may also deviate from their NE strategies \cite{colman2003cooperation}. In this case, despite being theoretically valid, the error-free beliefs, $\delta^{1*}$ and $\rho^{1*}$ defined in Definition 1 may not be consistent with the players' actual emotional strategies. As such, the rational strategies, $\alpha^*$ and $\beta^*$, that maximize the players' $\delta^{1*}$-based and $\rho^{1*}$-based perceived expected payoffs, may not maximize the players real expected payoffs  when the opponent deviates from its fully rational strategies, due to behavioral factors \cite{battigalli2009dynamic}. In addition, the soldier's and attacker's subjective emotions may also modify their objective functions to incorporate additional subjective goals. These psychological facets of the player's behavior in an IoBT network, are studied next using the framework of a dynamic psychological game \cite{battigalli2009dynamic}.  
 

 \section{Dynamic Psychological Game}\label{se:system}
 The formulated dynamic game in Section II captures the primary objectives of the soldier and attacker and the interdependence between these objectives. In this respect, in this game, each player, by using a set of beliefs about the opponent's strategy, aims at computing its optimal strategy to maximize its respective expected utility. Hence, the beliefs are considered to be solely a means using which a player can estimate its opponent's strategy in order to choose its optimal strategy, but are not considered a part of the utility function of each player. However, given the psychological (i.e. human) nature of the players in our game, their expectations, beliefs, and emotions have a direct effect on the way they perceive the outcome of the game. Indeed, by not achieving their expected (or belief-based) payoff, the soldier or attacker will experience frustration or anger, which has a direct impact on the way they assess and perceive the outcome of the game. In addition, due to the adversarial nature of the relationship between the soldier and attacker, in addition to achieving their own objective by maximizing/minimizing the communication delay, each may also strive to intentionally hurt the opponent, by aiming at frustrating the opponent or, more generally, causing a psychological (i.e. emotional) damage to this opponent. Hence, incorporating this psychological aspect in the formulation of the utility functions of each player enables a more general and representative game analysis that can realistically capture the psychological decision making processes and behavior of each of the soldier and attacker.

To this end, we next incorporate notions from psychological games~\cite{{battigalli2009dynamic}} in our game formulation to capture and analyze this psychological aspect of the decision making processes of the attacker and soldier. As such, in our introduced psychological game, the players expectations and beliefs will now be directly incorporated in their utility functions. In addition, given their objective to frustrate and anger the opponent, each player aims at anticipating the payoff that the opponent expects. To this end, in addition to building beliefs over the opponent strategies, each player also aims at building a belief system over the opponent's beliefs. This would, hence, enable anticipating the expectations of the opponent and, as a result, maximize its frustration. 

   \subsection{Psychology in the battlefield}
 
     In the aforementioned IoBT scenario, when one player (i.e. the soldier or the attacker) chooses its strategy such that its opponent receives a material payoff lower than expected, this player successfully frustrates its opponent. For example, if the soldier believes that the attacker did not launch an attack on the IoBT network at step $x$, it will communicate with this IoBT device $x$ and expect to achieve a material payoff $\overline \pi$. If, in reality, the attacker attacked $x$, then the material payoff of the soldier will decrease to $ \tilde\pi$. The gap between $\overline \pi$ and $ \tilde\pi$ quantifies the soldier's frustration. Note that the soldier and attacker only feel frustrated when they get a lower material payoff, compared to their expected material payoff. Hence, in our proposed psychological game formulation, the soldier and attacker will intentionally attempt to frustrate each other while also seeking to achieve their own, individual objectives. Ultimately, the soldier and attacker's intention to frustrate each other, combined with their individual objectives  (i.e. to minimize or maximize the soldier's communication delay), will determine the soldier and attacker's strategies in the battlefield.
   
     To consider their opponents' frustration in their own payoffs, the soldier and attacker should estimate their opponent's expected payoffs. This estimation requires the soldier and the attacker to build beliefs about their opponent's first-order beliefs, i.e. to build \emph{second-order beliefs}.
 The soldier's second-order belief on the attacker's first-order belief  at history $h^x$, $\boldsymbol{\rho}  ^1_{{h}^x}$, is denoted by a vector $\boldsymbol{\delta}^2_{{h}^x}  = \left[ {{\delta ^2_{{h}^x}}\left( 1 \right),{\delta^2_{{h}^x}}\left( 2 \right)} \right]$. The attacker's second-order belief on the soldier's belief $\boldsymbol{\delta} ^1_{{h}^x}$ is denoted by $\boldsymbol{\rho} ^2_{{h}^x} = \left[ {\rho^2_{{h}^x}\left( 1 \right),\rho^2_{{h}^x}\left( 2 \right)} \right]$.  Hereinafter, we use ${\delta}^2=\left\{\boldsymbol{\delta}  ^2_{\mathcal{H}^1}, \cdots, \boldsymbol{\delta}  ^2_{\mathcal{H}^X}\right\}$ to denote the set of soldier's second-order beliefs on the attacker, and ${\rho}^2=\left\{\boldsymbol{\rho}  ^2_{\mathcal{H}^1}, \cdots, \boldsymbol{\rho}  ^2_{\mathcal{H}^X}\right\}$ to denote the set of attacker's second-order beliefs on the soldier, for each possible history.

  \subsection{Soldier and attacker's frustration}

We define the soldier's and attacker's frustration as the gap between their expected material payoffs, respectively defined in (11) and (13), and their actual material payoffs. This frustration, indeed, stems from the fact that the soldier (attacker) may choose an action $a_n \in \mathcal{A}_{h^x}$ ($b_m\in \mathcal{B}_{h^x}$) that may be different from what the attacker (soldier) has anticipated based on its belief $\rho^1_{{h}^x}$ ($\delta^1_{{h}^x}$). Thus, in the considered IoBT network, under terminal history $\hat{h}^X$, the soldier's frustration with strategy ${\alpha}$ and belief ${\delta}^1$ will be given by (given that the soldier aims at maximizing $\overline{\pi}$ defined in (11)):
 \begin{equation}
\begin {split}
F\left( {{\alpha}, {{\delta}}^1}, \hat{h}^X \right) = \left[ {\overline{\pi}} \left( {{{\alpha},{\delta}^1}} \right)-{\tilde{\pi}} \left( {{\hat{h}^X}} \right) \right]^+,
\end {split}
 \end{equation}
 where $\left[x\right]^+= {\max } \left\{0,x\right\}$. ${\tilde{\pi}} \left( {{\hat{h}^X}} \right)=\frac{\Delta-{\tau {\left({\hat{h}}^X\right)}  }}{\Delta}$ is the soldier's actual payoff under terminal history $\hat{h}^X$.   
 Note that the attacker has no knowledge of the soldier's first-order belief, ${\delta}^1$, and strategy, $\alpha$. Hence, based on its sets of first-order and second-order beliefs, ${\rho}^1$ and ${\rho}^2$, on the soldier's strategy ${\alpha}$ and first-order belief ${\delta}^1$, the attacker can form a belief-based perception of the soldier's frustration, denoted by $F_a(\rho^1,\rho^2,\hat{h}^X)$, when a terminal history $\hat{h}^X$ occurs, is expressed as follows:
 \begin{equation}
\begin{split}
&F_a\left( {  {{\rho}}^1, {{\rho}}^2, \hat{h}^X}\right) = \left[ {\overline{\pi}} \left( {{{{\rho}}^1,{{\rho}}^2}} \right)-{\tilde{\pi}} \left( \hat{h}^X \right) \right]^+,
\end{split}
\end{equation}
where ${\overline{\pi}} \left( {{{{\rho}}^1,{{\rho}}^2}} \right)=\frac{\Delta-\sum\limits_{{\hat{h}^X} \in \mathcal{Z}} Q_{{{\rho}}^1,{{\rho}}^2}{\left({\hat{h}}^X\right)} {\tau {\left({\hat{h}}^X\right)}  }}{\Delta}$. In addition, $Q_{{{\rho}}^1,{{\rho}}^2}{\left({\hat{h}}^X\right)} $ is the attacker's perceived belief-based probability of occurrence of the terminal histories $\hat{h}^X \in \mathcal{Z}$ induced by its first-order and second-order beliefs, ${\rho}  ^1$ and ${{\rho}}^2$:
   \begin{equation}
\begin{split}
{Q_{{{\rho}}^1,{{\rho}}^2 }}\left( {{{\hat{h}}^X}} \right) = &\mathop \prod \limits_{x = 1}^{X} \left( {{{{\rho}}^1 _{{h^x}}}\left( 1 \right)  {\mathbbm{1}_{a\left( h^x \right) = {a_1}}} + {{{\rho}}^1 _{{h^x}}}\left( 2 \right) {\mathbbm{1}_{a\left( h^x \right) = {a_2}}}} \right)\\
&\times \left( {{{{\rho}}^2 _{{h^x}}}\left( 1 \right) {\mathbbm{1}_{b\left( h^x \right) = {b_1}}} + {{{\rho}}^2 _{{h^x}}}\left( 2 \right) {\mathbbm{1}_{b\left( h^x \right) = {b_2}}}} \right).
\end{split}
 \end{equation}
 
 
Thus, combining the attacker's primary objective of maximizing the soldier's communication delay at a minimum needed total jamming power with its intention to frustrate the soldier results in the following {belief-based expected psychological payoff (i.e. belief-based expected psychological utility)}:
 \begin{equation}\small
\begin{split}
u'&\left( { {{\beta}, {\rho}}^1,{{\rho}}^2 } \right) \\
&= \overline\pi'\left( {{{\beta}, {\rho}  ^1}} \right) +{\omega  _a} \sum\limits_{{\hat{h}^X} \in \mathcal{Z}} Q_{{\rho^1}, {\beta}}{\left({\hat{h}}^X\right)} F_a\left( {  {{\rho}}^1, {{\rho}}^2,\hat{h}^X}\right),
\end{split}
 \end{equation}
 where $\omega_{a}\in\left[0,1\right]$ is a parameters that represents the attacker's motivation  and willingness to frustrate the soldier.

 Similarly, under history $\hat{h}^X$, the frustration of the attacker with strategy ${\beta}$,  under the first-order belief ${{\rho}}^1$, will be:
  \begin{equation}
\begin {split}
&F'\left({\beta},  {\rho}  ^1,  \hat{h}^X\right) = \left[\overline\pi'\left( {{{\beta}, {\rho}  ^1}} \right)  -  \pi'_0\left( \hat{h}^X \right)\right]^+,
\end{split}
 \end{equation}
 where $\pi'_0\left( \hat{h}^X \right)$ is the attacker's actual payoff at terminal history, as defined in (9). 
 Based on its first-order and second-order beliefs, ${\delta} ^1$ and ${\delta} ^2$, the soldier can form a belief-based perception of qualify the attacker's frustration, denoted by $F_s(\delta^1, \delta^2,\hat{h}^X)$, when terminal history $\hat{h}^X$ is achieved, as follows:
\begin{equation}
\begin{split}
&F_s \left(  {\delta} ^1, {\delta} ^2,\hat{h}^X  \right) = \left[\overline\pi'\left( {{{\delta} ^1, {\delta} ^2}} \right)  -  \pi'_0\left( \hat{h}^X\right)\right]^+,
\end{split}
\end{equation}
where $\pi'\left( {{{\delta} ^1,{\delta} ^2}} \right)=\sum\limits_{{\hat{h}^X} \in \mathcal{Z}} Q_{{\delta} ^2,{\delta} ^1}\left({\hat{h}}^X\right) {\pi'_0\left({\hat{h}}^X\right)  }$. Here, $Q_{{\delta} ^2,{\delta} ^1}\left({\hat{h}}^X\right)$ is the  perceived belief-based probability of occurrence of terminal history $\hat{h}^X \in \mathcal{Z}$ based on the soldier's first-order and second-order beliefs, ${\delta}  ^1$ and ${{\delta}}^2$, and is defined as:
   \begin{equation}
\begin{split}
{Q_{{\delta} ^2,{\delta} ^1}}\left( {{{\hat{h}}^X}} \right) = &\mathop \prod \limits_{x = 1}^{X} \left( {{{\delta} ^2 _{{h^x}}}\left( 1 \right)  {\mathbbm{1}_{a\left( h^x \right) = {a_1}}} + {{\delta} ^2 _{{h^x}}}\left( 2 \right) {\mathbbm{1}_{a\left( h^x \right) = {a_2}}}} \right)\\
&\times\left( {{{\delta} ^1 _{{h^x}}}\left( 1 \right) {\mathbbm{1}_{b\left( h^x \right) = {b_1}}} + {{\delta} ^1 _{{h^x}}}\left( 2 \right) {\mathbbm{1}_{b\left( h^x \right) = {b_2}}}} \right).
\end{split}
 \end{equation}

Then, the soldier's goal to minimize its expected communication delay combined with its intention to frustrate the attacker can be captured by the following {belief-based expected psychological payoff (i.e. belief-based expected psychological utility)}:  
   \begin{equation}\small
\begin {split}
u& \left({\alpha}, {\delta} ^1,{\delta} ^2 \right) \\
&={\overline{\pi}} \left( {{{\alpha},{\delta}^1}} \right) +{\omega _{s}}\sum\limits_{{\hat{h}^X} \in \mathcal{Z}} Q_{{\alpha}, {\delta}^1}{\left({\hat{h}}^X\right)}  F_s \left( {\delta} ^1, {\delta} ^2,\hat{h}^X  \right),
\end{split}
 \end{equation}
where $\omega_{s}\in\left[0,1\right]$ is a parameter that represents the soldier's motivation to frustrate the attacker. 

  
\subsection{Dynamic Psychological game}
To capture the decision making processes of of the soldier and attacker, we introduce a \emph{dynamic psychological game} $[\mathcal{P},\mathcal{H}, \mathcal{Z}, u, u']$, 
where, similarly to the dynamic game defined in Section II-E, $\mathcal{P}$ is the set of players including the soldier and attacker. 
$\mathcal{H}$ is the set of histories, 
 and $\mathcal{Z}$ is the set of terminal histories.  In addition, $u$ and $u'$ represent the soldier's and the attacker's psychological expected utility defined in (26) and (22), respectively. In this psychological game, the soldier and the attacker aim at maximizing their belief-based psychological expected utilities. In this regard, when the first-order and second-order beliefs of each player correctly predict the strategy and the first-order belief of the opponent, and when each player chooses a strategy that maximizes its belief-based psychological expected utility based on those correct beliefs, these strategies give rise to a \emph{psychological equilibrium (PE)} \cite{{battigalli2009dynamic}}. In this respect, the PE of our proposed psychological game is formally defined as follows:

\begin{definition}\emph{The \emph{psychological equilibrium} of the formulated psychological game is defined as $\left({\alpha}^*, {\beta}^*, {\delta}^{1*}, {\delta}^{2*}, {\rho}^{1*}, {\rho}^{2*}\right)$,  in which ${\alpha}^*$, and ${\beta}^*$ are rational, such that: 
\begin{equation}
{\alpha}^* \in \mathop {\arg \max }\limits_{{\alpha} \in \mathcal{C}} u \left( {\alpha}, {\delta} ^{1*},{\delta} ^{2*}, \right),
\end{equation}
\begin{equation}
{{\beta}}^* \in \mathop {\arg \max }\limits_{{\beta} \in \mathcal{D}} u'\left( {{\beta}, {{\rho}}^{1*},{{\rho}}^{2*}, } \right),
\end{equation}
while  the first-order and second-order beliefs, ${\delta}^{1*}$, ${\delta}^{2*}$, ${\rho}^{1*}$, and ${\rho}^{2*}$, are error-free such that 
for all ${a}_n \in \mathcal{A}_{h^x}$ at each history ${h}^x$: 
\begin{equation}
{{\delta}^{2*}_{{h}^x} }\left( n \right) ={{\rho}_{{h}^x} ^{1*}}\left(n \right) ={\alpha}^*_{{h}^x}\left(n \right),
\end{equation} 
and for all $b_m \in \mathcal{B}_{h^x}$ at each history ${h}^x$:
\begin{equation}
{{\rho}_{{h}^x} ^{2*}}\left(m \right) ={{\delta}^{1*}_{{h}^x}}\left( m \right) = {\beta}^{*}_{h^x}\left( m \right).
\end{equation}}
\end{definition}

The principal difference between the PE and the NE (which is an underlying difference between the proposed standard dynamic game and the proposed psychological game) is that the utility function of each player is not only dependent on the strategy or action chosen by the opponent, but also on the opponent's beliefs. In this regard, the payoff of each player in the psychological game does not only depend on what the opponent does, but also on what the opponent thinks. This enlarges the domain of analysis of the game to incorporate psychological aspects of the players' decision making processes, which are not typically present in a traditional dynamic game formulation.  Hence, even though the definition of the PE still requires that the first-order and second-order beliefs of each player are error-free, since these beliefs are incorporated in the payoffs of each player, they will have a direct effect on their rationally chosen (i.e. PE) strategies. In essence, at a PE, the players' intention to frustrate one another is captured, as the soldier and the attacker make rational determination on their strategies to maximize both their belief-based expected material payoff and their opponents' frustration, based on their error-free first-order and second-order beliefs. 
Based on \cite{battigalli2009dynamic}, there always exists at least one such PE in the formulated psychological game. In particular, under the assumptions that i) evaded attacks at history $h^x$ yield higher expected material payoffs for the soldier and lower expected material payoffs for the attacker, at the current and future histories, and ii)  a successful (unjammed) communication at history $h^x$ yields a higher expected material payoff for the soldier and a lower expected material payoff for the attacker, at the current and future histories, we can derive Theorem 1 and Theorem 2:

\begin{theorem}\label{theorem1}
\emph{The NE and the PE of, respectively, the conventional dynamic game and the psychological game are unique. }

\end{theorem}

\begin{proof} 
At history ${h}^x\in\mathcal{H}^x$, the soldier chooses its action from $\left\{ {{a}_1},{{a}_2} \right\}$, while the attacker chooses its action from $\left\{ {{b}_1},{{b}_2} \right\}$. We represent the soldier's and attacker's payoffs when the soldier takes action $a_n$ and the attacker takes action $b_m$, where $n,m\in\{1,2\}$, by $\pi_{n,m}$ and $\pi'_{n,m}$, respectively. Here, $\pi_{n,m}$ and $\pi'_{n,m}$ include the instantaneous payoffs the soldier and attacker receive when taking their action pair at history $h^x$ as well as future expected payoffs at the following histories. As such, under each combination of the soldier's and attacker's pure strategies at $h^x$, the soldier's possible payoffs are represented by  ${{\pi}}_{1,1}$, ${{\pi}}_{1,2}$, ${{\pi}}_{2,1}$ and ${{\pi}}_{2,2}$, while the attacker's payoffs are represented by ${{\pi}'}_{1,1}$, ${{\pi}'}_{1,2}$, ${{\pi}'}_{2,1}$ and ${{\pi}'}_{2,2}$.
 Note that, here, we consider that the following inequalities hold: ${{\pi}}_{1,2}>{{\pi}}_{2,2} \ge {{\pi}}_{1,1}$, and ${{\pi}}_{2,1}>{{\pi}}_{2,2}\ge{{\pi}}_{1,1}$. Indeed, $\pi_{1,2}>\pi_{2,2}$ reflects the gain that the soldier receives from communicating with the IoBT network without being jammed by the attacker, while $\pi_{2,2}\geq\pi_{1,1}$ reflects the loss the attacker incurs from attempting to communicate with a jammed IoBT network. In addition, $\pi_{2,1}>\pi_{2,2}$ reflects the gain the soldier will receive in future steps due to the attacker wasting some of its jamming power when the soldier had not attempted to communicate with the IoBT network. Similarly, we also consider the following inequalities to hold:  ${{\pi}'}_{1,1}\ge{{\pi}'}_{2,2}>{{\pi}'}_{1,2}$, ${{\pi}'}_{1,1}\ge{{\pi}'}_{2,2}>{{\pi}'}_{2,1}$. These inequalities correspond to considering that:  i) evaded attacks at history $h^x$ yield higher expected material payoffs for the soldier and lower expected material payoffs for the attacker, at the current and future histories, and ii)  a successful (unjammed) communication at history $h^x$ yields a higher expected material payoff for the soldier and a lower expected material payoff for the attacker, at the current and future histories. 
In this respect, we can compute the psychological payoff of the soldier and the attacker under each combination of these pure strategies. In this regard, we consider $\alpha'$ to be the attacker's belief representing the probability with which the attacker believes that the soldier will choose action $a_1$. In addition, we consider $\beta'$ to be the belief that the soldier has, representing the probability with which the soldier believes that the attacker will choose action $b_1$. Under the correctness of beliefs defined in (29) and (30) of the PE, these probabilities also reflect the second-order beliefs of the players as well as the actual strategies chosen by each of the players.  Starting from the soldier's side, the soldier's psychological payoffs at the strategy pairs $\left\{a_1,b_1\right\}$ and $\left\{a_1,b_2\right\}$ are, respectively, $\pi_{1,1}$ and ${{\pi}}_{1,2} +\omega_s \left(1-\alpha'\right)\left({{\pi}'}_{2,2}-{{\pi}'}_{1,2}\right)$. In addition, psychological payoffs of the soldier at the pure strategy pairs $\left\{a_2,b_1\right\}$ and $\left\{a_2,b_2\right\}$ are, respectively, ${{\pi}}_{2,1} +\omega_s \alpha'\left({{\pi}'}_{1,1} -{{\pi}'}_{2,1} \right)$ and ${{\pi}}_{2,2}$.
 On the other hand, the attacker's psychological payoffs at strategy pairs $\left\{a_1,b_1\right\}$ and $\left\{a_1,b_2\right\}$ are, respectively, ${{\pi}'}_{1,1}+\omega_a\left(1- \beta'\right)\left({{\pi}} _{1,2}-{{\pi}}_{1,1}\right)$ and ${{\pi}'}_{1,2}$. In addition, the attacker's psychological payoffs at strategy pairs $\left\{a_2,b_1\right\}$ and $\left\{a_2,b_2\right\}$ are, respectively, ${{\pi}'}_{2,1}$ and ${{\pi}'}_{2,2} +\omega_a \beta'\left({{\pi}}_{2,1} -{{\pi}}_{2,2}\right)$.

 In the conventional dynamic game, in which the frustrations of the soldier and attacker are not considered in their opponent's utility functions, we denote the soldier's and attacker's strategies by $\left[\alpha,1-\alpha\right]$ and $\left[\beta,1-\beta\right]$, respectively.  By using the indifference principle,  we can compute the NE strategy of the soldier, which results in $\alpha =\frac{{\pi '_{2,2}- \pi '_{2,1}}}{{\pi '_{1,1} + \pi '_{2,2} - \pi '_{1,2} - \pi'_{2,1}}} $, and the NE strategy of the attacker, which results in $\beta=\frac{{\pi_{2,2} - \pi_{1,2}}}{{\pi_{1,1} + \pi_{2,2}- \pi_{1,2} - \pi_{2,1}}}$.
  Here, we note that this computed NE is unique since it can be shown that no NE exists in pure strategies, under our considered set of inequalities defined at the start of the proof, and the solution of the equations resulting from the indifference principle results in unique mixed-strategies $\alpha$ and $\beta$.

   Now, considering the psychological game, we denote the soldier's and the attacker's strategies at the PE by $\left[\alpha',1-\alpha'\right]$ and $\left[\beta',1-\beta'\right]$, respectively.  By using the indifference principle, with the soldier and the attacker holding correct (i.e. error-free) beliefs on one another, we can compute the PE strategies as follows:
  \begin{equation}\label{eq:31}
  \begin{split}
  \alpha '& =\\
     &\frac{{{\omega _a}\beta '\left( {\pi_{2,1} - \pi_{2,2}} \right) + \left( {\pi '_{2,2} - \pi '_{2,1}} \right)}}{{D' + {\omega _a}\left[ { \beta '\left(\pi_{2,1}-\pi_{2,2}\right) + \left( {1 - \beta '} \right)\left(\pi_{1,2} - \pi_{1,1}\right)} \right]}},
   \end{split}
   \end{equation}
 \begin{equation}\label{eq:32}
  \begin{split}
 \beta ' &= \\
   &\frac{{{\omega _s}\left( {1 - \alpha '} \right)\left( {{{\pi '}_{2,2}} - {{\pi '}_{1,2}}} \right) + \left( {{\pi _{1,2}} - {\pi _{2,2}}} \right)}}{{D + {\omega _s}\left[ {\left( {1 - \alpha '} \right)\left( {{{\pi '}_{2,2}} - {{\pi '}_{1,2}}} \right) + \alpha '\left({{\pi '}_{1,1}}- {{\pi '}_{2,1}}\right)} \right]}}.
   \end{split}
   \end{equation}
where $D=\pi_{1,2}+\pi_{2,1}-\pi_{1,1}-\pi_{2,2}$, $D'=\pi'_{1,1}+\pi'_{2,2}-\pi'_{1,2}-\pi'_{2,1}$. Note that, in (\ref{eq:31}), when $\beta'=0$, $\alpha'<\alpha$, when $\beta'=1$, $\alpha'>\alpha$. At the same time, in (\ref{eq:32}), when $\beta'=0$, $\alpha'=1+\frac{\pi_{1,2}-\pi_{2,2}}{\omega_s(\pi'_{2,2}-\pi'_{1,2})}>\alpha$, when $\beta'=1$, $\alpha'=\frac{\pi_{1,1}-\pi_{2,1}}{\omega_s\left(\pi'_{1,1}-\pi'_{2,1}\right)}<\alpha$. 
Hereinafter, we rewrite (\ref{eq:31}) as:

  \begin{equation}\label{eq:33}
  \setlength{\abovedisplayskip}{0 pt}
\setlength{\belowdisplayskip}{3 pt}
  \begin{split}
  &\alpha ' =\frac{{F_3\beta ' + F_4}}{{F_1\beta' +F_2}},
   \end{split}
   \end{equation}
where $F_1=\omega_a\left(\pi_{2,1}+\pi_{1,1}-\pi_{1,2}-\pi_{2,2}\right)$, $F_2=\pi '_{2,2} - \pi '_{2,1}+\pi '_{1,1} - \pi '_{1,2}+\omega_a\left(\pi_{1,2}-\pi_{1,1}\right)>0$, $F_3=\omega_a\left(\pi_{2,1}-\pi_{2,2}\right)>0$, $F_4=\pi '_{2,2} - \pi '_{2,1}>0$. Meanwhile, in (\ref{eq:33}), $\frac{\partial{\alpha'}}{\partial{\beta'}}=\frac{F_2F_3-F_1F_4}{\left(F_1\beta'+F_2\right)^2}$. Note that, when $\pi'_{2,2}-\pi'_{2,1}\ge\theta_1\left(\pi_{2,1}-\pi_{2,2}\right)$, we can prove that $F_2F_3-F_1F_4>0$, which implies that, in (\ref{eq:31}), $\alpha'$ increases with an increase in $\beta'\in[0,1]$.

Meanwhile, we rewrite equation (\ref{eq:32}) with:
     \begin{equation}\label{eq:34}
\setlength{\belowdisplayskip}{3 pt}%
  \begin{split}
  &\alpha ' = \frac{{F'_3\beta ' + F'_4}}{{F'_1\beta' +F'_2}},
   \end{split}
   \end{equation}
    where $F_1'=\omega_s\left(\pi'_{1,2}+\pi'_{1,1}-\pi'_{2,1}-\pi'_{2,2}\right)$, $F'_2=\omega_s\left(\pi'_{2,2}-\pi'_{1,2}\right)>0$, $F'_3=-D-\omega_s\left(\pi'_{2,2}-\pi'_{1,2}\right)<0$, $F'_4={\pi _{1,2}}-{\pi _{2,2}}+\omega_s\left(\pi '_{2,2} - \pi '_{1,2}\right)>0$. In (\ref{eq:34}), $\frac{\partial{\alpha'}}{\partial{\beta'}}=\frac{F'_2F'_3-F'_1F'_4}{\left(F'_1\beta'+F'_2\right)^2}$. Here, $F'_2F'_3-F'_1F'_4\le\omega_s\left(\pi'_{2,2}-\pi'_{1,2}\right)\left(\pi_{1,1}-\pi_{2,1}\right)<0$.  
    Thus $\frac{\partial{\alpha'}}{\partial{\beta'}}<0$, which implies that, in (\ref{eq:32}), $\alpha'$ decreases with an increase in $\beta'\in[0,1]$.

In conclusion, when $\beta'=0$, $\alpha'$ in (\ref{eq:31}) is smaller than $\alpha'$ in (\ref{eq:32}), while when $\beta'=1$, $\alpha'$ in (\ref{eq:31}) is larger than $\alpha'$ in (\ref{eq:32}). In (\ref{eq:31}), $\alpha'$ is strictly increasing in $\beta'$, while, in (\ref{eq:32}), $\alpha'$ is strictly decreasing in $\beta'$. As such, for $0\le\beta'\le1$, (\ref{eq:31}) and (\ref{eq:32}) has $1$ intersection point. This implies that the solution obtained from the indifference principle is unique and there is a unique PE in mixed strategies. In addition, given the inequalities stated at the beginning of the proof, it can be readily shown that no PE exists in pure strategies. Therefore, under the considered set of inequalities, the PE of the game is unique. This completes the proof.
 \end{proof}
The incorporation of the opponent's beliefs in the objective function of each player, and the possible willingness of each player to not only meet its own objective but to frustrate the opponent, introduce significant modifications to the equilibrium strategies of each player  as shown in Theorem 2. 

 \begin{theorem}\label{theorem2}
\emph{In the psychological game, at the PE, the attacker is more likely to target the IoBT device having the best channel conditions as compared to the NE of the traditional dynamic game.} 

\end{theorem}

\begin{proof} 
  In (\ref{eq:32}), when $\beta'=\beta$, we can get:
       \begin{equation}\label{eq:35}
\setlength{\belowdisplayskip}{3 pt}%
  \begin{split}
 &\alpha'=\\
 &\small{\frac{\omega_s\left(\pi'_{2,2}-\pi'_{1,2}\right)\left(\pi_{2,1}-\pi_{1,1}\right)}{\omega_s\left[\left(\pi'_{2,2}-\pi'_{1,2}\right)\left(\pi_{2,1}-\pi_{1,1}\right)+\left(\pi'_{1,1}-\pi'_{2,1}\right)\left(\pi_{1,2}-\pi_{2,2}\right)\right]}}
   \end{split}
   \end{equation}
 such that $\alpha'\ge\frac{1}{2}$, since $\pi'_{1,1}-\pi'_{2,1}\le\theta_1\left(\pi_{2,1}-\pi_{1,1}\right)$, $\pi'_{2,2}-\pi'_{1,2}\ge\theta_1\left(\pi_{1,2}-\pi_{2,2}\right)$.
In (\ref{eq:31}), when $\beta'=\beta$, $\alpha'<\frac{1}{2}$ if $\pi_{2,1}<{\pi_{1,2}}$. As such, if $\pi_{2,1}<{\pi_{1,2}}$ (i.e. if the current IoBT device exhibits the best channel as compared to all future IoBT devices), $\beta'>\beta$, at the intersection of (\ref{eq:31}) and (\ref{eq:32}), as shown in Fig. \ref{figure2}. Indeed, $\pi_{1,2}>\pi_{2,1}$ reflects the gain that the soldier receives from communicating with the current IoBT device is larger than the gain it may potentially receive if it successfully communicates with the IoBT network in future steps. Thus, when there is no remaining IoBT devices in future steps, which have a better channel  as compared to the current device, the attacker will be more likely to attack the current device. This completes
the proof.
 \begin{figure}[!t]
  \begin{center}
   \vspace{0cm}
    \includegraphics[width=7.5cm]{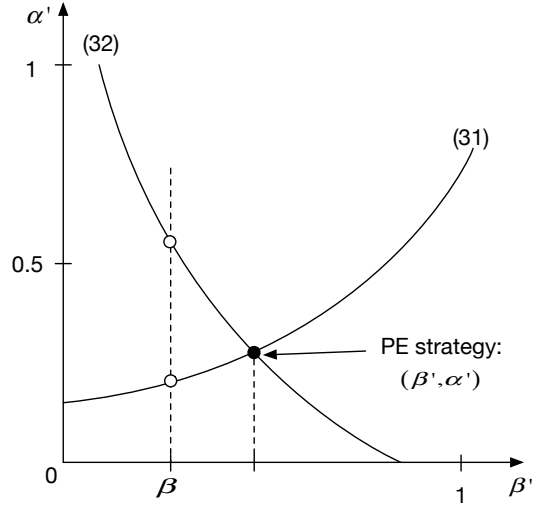}
    \vspace{-0.2cm}
    \caption{\label{figure2} $\alpha'$ versus $\beta'$.}
  \end{center}
\end{figure}
 \end{proof}
Hence, Theorem 2 shows the effect that the incorporation of beliefs in each player's objective function (as in the proposed psychological game) can have on their equilibrium strategies. In particular, even when the beliefs are error-free, the fact that the beliefs have a direct effect on how the outcome of the game is assessed by each player has a direct impact on the chosen equilibrium strategies. Indeed, Theorem 2 shows that, for the attacker, incorporating a belief over the soldier's strategy and a belief over the soldier's belief in its objective function allows the attacker to modify its equilibrium strategy in order to maximize the soldier's frustration.

 In summary, the formulated psychological game enables analysis of the soldier's and attacker's psychological intention to frustrate each other  and allows studying the effect of such a psychological decision making aspect on the chosen strategies of each player.  As can be seen from Definition 2 and conditions (29) and (30), holding correct beliefs is necessary to reach a PE of the game. This highlights the importance of the beliefs and the effect that they have on the chosen soldier and attacker strategies. In practical applications, in which each player may not have the ability to analytically characterize a set of error-free beliefs, a learning algorithm could be applied to numerically synthesize these beliefs \cite{chen2016caching,alpaydin2009introduction}. To this end, a Bayesian updating based algorithm is proposed next, which enables a numerical computation of first-order and second-order beliefs of each player, which as a result, allows computation of equilibrium strategies of the psychological game.
\section{Learning to Be Rational in the Psychological Game} \label{se:3}  
In the studied IoBT scenario, the soldier and attacker form first-order and second-order beliefs, which are used to compute their optimal (i.e. rational) strategy. As such, as shown in Definition 2, the PE strategies require the players to form correct (i.e. error-free) beliefs. Forming such beliefs analytically is a typically complex task especially when the size of the battlefield grows, which leads to having a significantly large number of possible histories in the game. Hence, rather than relying on complex analytical derivations, the soldier and attacker can rely on numerical techniques and observations to form such beliefs, and as a result, choose their strategies. A powerful tool which can be used to form the players' beliefs in our proposed psychological game is \emph{Bayesian updating} \cite{jaffray1992bayesian}, which enables the use of observations to form consistent beliefs. Hence, next, we develop a  Bayesian updating based approach to solve the proposed psychological game, by first predicting the players' future strategies and beliefs.

By using Bayesian updating, the attacker and the soldier find, at each history $h^x$ in the game, find the posterior probabilities as follows:
 \begin{equation}\label{eq:36}
\Pr \left( {{a_n}\left| {{{h}}^x} \right.} \right) = \frac{{\Pr \left( {{{{h}}^x}\left| {{a_n}} \right.} \right)\Pr \left( {{a_n}} \right)}}{{\Pr \left( {{{h}}^x} \right)}},
\end{equation}  
 \begin{equation}\label{eq:37}
\Pr \left( {{b_m}\left| {{{h}}^x} \right.} \right) = \frac{{\Pr \left( {{{{h}}^x}\left| {{b_m}} \right.} \right)\Pr \left( {{b_m}} \right)}}{{\Pr \left( {{{{h}}^x}} \right)}},
\end{equation}  
where $\Pr \left( {{a_n}\left| {{{h}}^x} \right.} \right) $ represents the probability of the soldier choosing action $a_n$ at history ${{{h}}^x}$, while $\Pr \left( {{b_m}\left| {{{h}}^x} \right.} \right)$ is the probability of the attacker choosing action ${b_m}$ at history ${{{h}}^x}$. $\Pr \left( {{{{h}}^x}\left| {{a_n}} \right.} \right)$ represents the probability that the current history is ${{{h}}^x}$ when the soldier chooses action $a_n$ at step $x$. $\Pr \left( {{{{h}}^x}\left| {{b_m}} \right.} \right)$ represents the probability that the current history is ${{{h}}^x}$ when the soldier chooses action $b_m$ at step $x$. $\Pr \left( {{a_n}} \right)$ and $\Pr \left( {{b_m}} \right)$, respectively, represent the probabilities that action $a_n$ or $b_m$ is chosen at step $x$. ${{\Pr \left( {{{h}}^x} \right)}}$ represents the probability that history ${{{h}}^x}$  is reached at step $x$. 

Under Bayesian updating, the soldier and the attacker build their first-order beliefs on their opponent's strategy, under each history $h^x$, according to (\ref{eq:36}) and (\ref{eq:37}).  Meanwhile, based on Definition 2, the soldier and attacker's second-order beliefs are consistent with their own strategies. Based on their beliefs, the soldier and attacker determine their optimal strategies that maximize their belief-dependent utilities in (22) and (26).  The specific process of our Bayesian updating based solution is represented in Algorithm 1. Note that, all of the aforementioned probabilities will be updated through the repetition of the game\footnote{In practical IoBT scenarios, there will be more than one soldier working on a same mission. The soldier and the attacker's beliefs, which are consistent with the posterior probabilities defined in (\ref{eq:36}) and (\ref{eq:37}), are updated through different soldiers' accomplishment of the same mission.}. At the beginning of algorithm, the soldier and attacker assume $\Pr \left( {{a_n}} \right)=\Pr \left( {{b_m}} \right)=\frac{1}{{2}}$ at each step $x$, as there is no reason for them to assume that their opponents has any preference on the choice of their actions. Similarly, $\Pr \left( {{{{h}}^x}\left| {{a_n}} \right.} \right)=\frac{1}{V_{a_n,{{{h}}^x}}}$, $\Pr \left( {{{{h}}^x}\left| {{b_m}} \right.} \right)=\frac{1}{V_{b_m,{{{h}}^x}}}$ are assumed at the beginning of the algorithm, with ${V_{a_n,{{{h}}^x}}}$, ${V_{b_m,{{{h}}^x}}}$ representing the total number of times that when $a_n$ or $b_m$ are respectively chosen, they are chosen from history $h^x$.  We also consider ${{\Pr \left( {{{h}}^x} \right)}}=1$ when history $h^x$ is reached at the first iteration of the game. 
%
 \begin{algorithm}[t]\footnotesize
\caption{Bayesian updating solution for the dynamic psychological game }   
\label{alg:Framwork}   
\setlength{\abovecaptionskip}{-15pt} 
\setlength{\belowcaptionskip}{-15pt}
\begin{algorithmic} [1] 
\REQUIRE The set of IoBT devices ${\mathcal{X}}$ in the battlefield, the number of required communication links between the soldier and IoBT devices $J$, and the power limitation of the attacker $E$. \\ 
\vspace{2pt}  
\ENSURE Initialize belief of the soldier and attacker.\\

\vspace{2pt}  
\FOR {IoBT device $1$ to $X$} 
\vspace{2pt}  
\STATE Update ${{\Pr \left( {{{h}}^x} \right)}}$.
\vspace{2pt}  
\STATE Calculate $\Pr \left( {{a_n}\left| {{{h}}^x} \right.} \right)$, $\Pr \left( {{b_m}\left| {{{h}}^x} \right.} \right)$, with ${\delta^{2*}_{h^x} }\left( n \right) ={\rho_{h^x} ^{1*}}\left(n \right) = \Pr \left( {{a}_n\left| {{{h}}^x} \right.} \right)$, ${\rho_{h^x} ^{2*}}\left(m \right) ={\delta^{1*}_{h^x}}\left( m \right) = \Pr \left( {b_m\left| {{{h}}^x} \right.} \right)$.
\vspace{2pt}  
\FOR {the attacker}
\vspace{2pt}  
\FOR {all action ${{b}_m \in \mathcal{B}_{ {{h}}^x}}$} 
\vspace{2pt}  
\STATE Estimate the soldier expected material payoff based on ${\rho}^{2*}_{h^x}$.
\vspace{2pt}  
\ENDFOR
\vspace{2pt}  
\STATE Find $b^* \left( {{h}}^x\right)= \mathop {\arg \max }\limits_{{b}_m \in \mathcal{B}_{ {{h}}^x}} u'$.
\vspace{2pt}  
\STATE Update ${{\Pr \left( b_m \right)}}$, ${\Pr \left( { {{h}}^x\left| {{b_m}} \right.} \right)}$, for all ${{a}_n \in \mathcal{A}_{ {{h}}^x}}$.
\vspace{2pt} 
\ENDFOR 
\vspace{2pt}  
\FOR {the soldier}
\vspace{2pt}  
\FOR {all action $a_n \in \mathcal{A}_{ {{h}}^x}$} 
\vspace{2pt}  
\STATE Estimate the attacker's expected material payoff based on ${\delta}^{2*}_{h^x}$.
\ENDFOR
\vspace{2pt}  
\STATE Find ${a}^*\left( {{h}}^x\right) = \mathop {\arg \max }\limits_{{a}_n \in \mathcal{A}_{ {{h}}^x}} u$.
\vspace{2pt}  
\STATE Update ${{\Pr \left( a_n \right)}}$, ${\Pr \left( {{{{h}}^x}\left| {{a_n}} \right.} \right)}$, for all $b_m \in \mathcal{B}_{{{h}}^x}$.
\vspace{2pt}  
\ENDFOR 
\vspace{2pt}  
\ENDFOR  
\end{algorithmic}
\end{algorithm} 

By extending the results in \cite{kalai1993rational} to our formulated psychological game, when the soldier and the attacker maximize their payoffs based on their beliefs learnt from Bayesian updating, their beliefs will always converge to a value that $\epsilon$-likes their opponents' strategies and beliefs. Here, a first-order belief at history $h^x$, $\boldsymbol{\delta}^1_{{h}^x}  = \left[ {{\delta ^1_{{h}^x}}\left( 1 \right),{\delta^1_{{h}^x}}\left( 2 \right)} \right]$, is said to \emph{$\epsilon$-like} a player's strategy (i.e. $\boldsymbol\beta_{{h}^x}=\left[\beta_1,\beta_2\right]$) when there exists an $\epsilon>0$, such that (i) $\sum\limits_{n = 1}^2\delta ^1_{{h}^x}\left( n \right)$ and $\sum\limits_{n = 1}^2\beta_n$ are greater than $1-\epsilon$; (ii) $\left(1-\epsilon\right)\beta_n\le\delta ^1_{{h}^x}\left( n \right)\le\left(1+\epsilon\right)\beta_n$, for $n \in \left\{1,2\right\}$. During the Bayesian updating process, the soldier and the attacker predict their opponent's strategies based on the sequence of actions that the soldier and attacker had taken during the updating process. 
In this regard, each player, including the soldier and the attacker, optimizes its utilities based on the beliefs learned from the actual strategies played in the game. Hence, learning a belief system over the strategy the opponent takes at a certain history requires this history to be reached during the learning process. Hence, if the sequence of actions taken by the soldier and attacker do not lead to a certain history to be reached, the player's cannot use previous observations to build a belief system over that particular history. As such, based on the Bayesian updating and optimal strategy selection, the soldier and the attacker are guaranteed to reach an $\epsilon$-like psychological self-confirming equilibrium (PSCE) defined as follow\cite{fudenberg1995learning}.
 
\begin{definition}\emph{The $\epsilon$-like \emph{PSCE} of the formulated dynamic psychological game is  defined as $\left({\alpha}^*, {\beta}^*, {\delta}^{1*}, {\delta}^{2*}, {\rho}^{1*}, {\rho}^{2*}\right)$,  where ${\alpha}^*$, and ${\beta}^*$ are rational; 
while ${\rho}^{1*}$, ${\delta}^{1*}$, ${\rho}^{2*}$, and ${\delta}^{2*}$ respectively $\epsilon$-like $\alpha^*$, $\beta^*$, ${\delta}^{1*}$, ${\rho}^{1*}$, \iffalse are consistent\fi for each history $h^x$ such that $\Pr \left( {{{\alpha}^*, {\beta}^*}\left| {{{h}}^x} \right.} \right)>0$.}
\end{definition}

As such, in an  $\epsilon$-like PSCE, the soldier and the attacker's beliefs, which are updated based on previous actions that have been taken in the battlefield, $\epsilon$-like the error-free beliefs, defined in the PE of this game. 
Hence, using the Bayesian updating process, the soldier and the attacker can dynamically predict their opponent's strategy. As such, rational strategies that maximize each player's payoff can be reached upon the convergence of the Bayesian updating algorithm.  

 \vspace{0.2cm}
\section{Simulation Results and Analysis} 
For our simulations, we consider a battlefield in which multiple IoBT devices are randomly distributed along the soldier's path. The channel gain between the soldier and each IoBT device follows a Rayleigh distribution with unit variance. 
The parameters used in the simulations are listed in Table  \uppercase\expandafter{\romannumeral1}. The \emph{Bayesian updating based PSCE} results, denoted by BU herinafter, are compared to the NE and PE results. 

\begin{table}
  \newcommand{\tabincell}[2]{\begin{tabular}{@{}#1@{}}#2\end{tabular}}
\renewcommand\arraystretch{1}
 \caption{
    \vspace*{-0.2em}SYSTEM PARAMETERS\cite{pal2011performance}}\vspace*{-1em}
\centering  
\begin{tabular}{|c|c|c|c|}
\hline
\textbf{Parameter} & \textbf{Value} & \textbf{Parameter} & \textbf{Value} \\
\hline
$P_S $ & 20 dBm & $P_A$ & 20 dBm\\
\hline
$J=J'$ & 1 & $\Delta$ & 80 ms \\
\hline
$\theta_1= \theta_2$ & 0.5 & $ \omega_s=\omega_a$ & 0.5\\
\hline
$I_x$ & 20 MHz & $ \sigma ^2$ & -95 dBm\\

\hline
\end{tabular}
\end{table}

\begin{figure}[!t]
  \begin{center}
    \includegraphics[width=9cm]{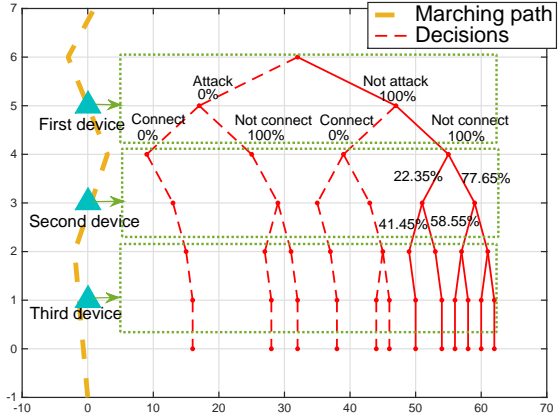}
    \caption{\label{figure3} The soldier and attacker's strategy in the battlefield.}
  \end{center}\vspace{-0.3cm}
\end{figure}

Fig. \ref{figure3} shows the way in which the soldier and attacker make decisions in an IoBT network with 3 devices, using the Bayesian updating algorithm. Fig. \ref{figure3} shows that, at each history in the game, the soldier (attacker) makes decisions on whether to communicate with (attack) the current device or not, based on their prediction on their opponent's strategies and beliefs. The soldier's belief on the attacker is updated based on the Bayesian updating algorithm. For example, at a certain history $h^x$ in the game, an IoBT device is attacked  $22.35\%\times M$ times in $M$ iterations of the Bayesian updating algorithm, where history $h^x$ is always reached at step $x$. As such, when history $h^x$ is reached in the $M+1$-th iteration of the Bayesian updating algorithm, the soldier will believe that the attacker will attack the current IoBT device with a probability $22.35\%$. Note that, the soldier updates its beliefs based on the  actual strategies played in the game. As the number of soldiers taking the same mission is limited, the opportunity with which the soldier updates its beliefs is limited. Hence, the soldier's beliefs in the battlefield  may not be error-free since not enough previous observations are available for the generated beliefs to $\epsilon$-like the correct strategies. In the following simulation results, the effect of the players' non-error-free beliefs will also be studied,  in addition to a complete numerical analysis investigating the NE, PE, and BU of the proposed games as well as studying the effects that the various game parameters have on these equilibria. For the following simulations, we consider 5 IoBT devices.

\begin{figure}[!t]
  \begin{center}
   \vspace{0cm}
    \includegraphics[width=8cm]{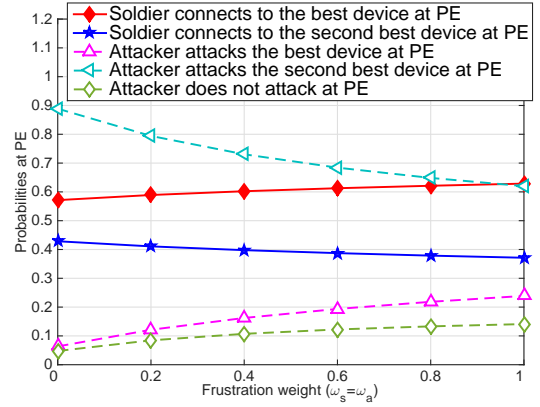}
    \vspace{-0.2cm}
    \caption{\label{figure4} The soldier and the attacker's PE strategies as the frustration weights vary.}
  \end{center}\vspace{-0.3cm}
\end{figure}

Fig. \ref{figure4} shows the effect that the weights of frustration of each player (i.e. $\omega_s$ and $\omega_a$) -- which reflects the importance each player assigns on frustrating the opponent -- have on the chosen PE strategies. In this regard, Fig. \ref{figure4} shows that, as the weight $\omega_a$ increases, the probability that the attacker attacks the best IoBT device  increases. Here, the ranking of best, second best device correspond to the channel quality of that device. Also, the probability that the attacker attacks the IoBT devices decreases with the increase of $\omega_a$, signifying that the attacker's likelihood of launching any jamming attack decreases. This is due to the fact that, attempting to frustrate the soldier, the attacker's strategy can be designed to increase the likelihood of compromising the best IoBT device along the path.
Meanwhile, since the attacker also aims at minimizing the total jamming power consumed, its equilibrium strategy will tend towards attacking the IoBT network with a lower probability. In addition, Fig. \ref{figure4} also shows that, as the weight $\omega_s$ increases, the probability that the soldier communicates with the IoBT device with best channel quality increases, as likelihood of occurrence of attacks decreases. Note that, when $\omega_s = 0$ and $\omega_a = 0$, the attacker's frustration is not considered in the soldier's payoff, while the soldier's frustration is not considered in the attacker's payoff, which corresponds to the NE of the conventional dynamic game introduced in Section II-E.

\begin{figure}
  \centering
  \subfigure[Convergence of the soldier's BU strategies]{
    \label{fig:subfig5:a} 
    \includegraphics[width=8cm]{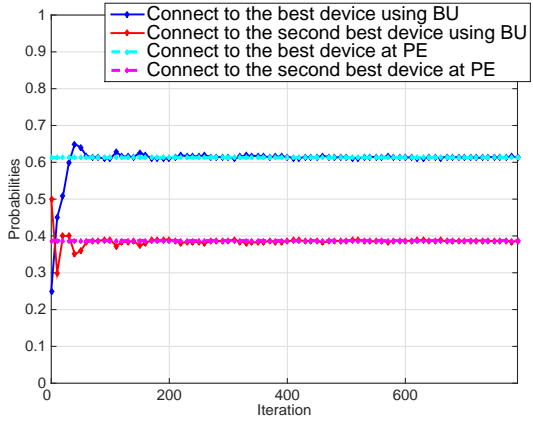}}
  \hspace{1in}
  \subfigure[Convergence of the attacker's BU strategies]{
    \label{fig:subfig5:b} 
    \includegraphics[width=8cm]{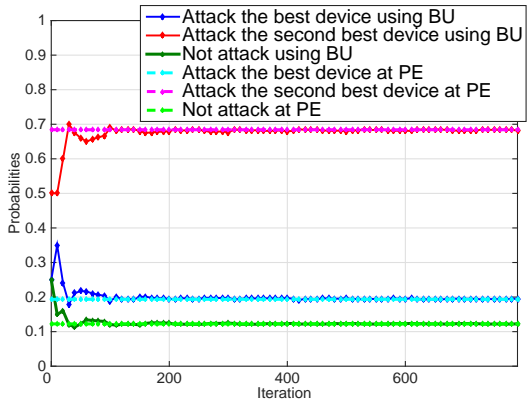}}
  \caption {Convergence of the BU strategies}
 \label{figure5}
 \vspace{-0.3cm} 
\end{figure}

Fig. \ref{figure5} shows the convergence of the soldier's and attacker's strategies using the proposed proposed Bayesian updating algorithm. Fig. \ref{fig:subfig5:a} and Fig. \ref{fig:subfig5:b} show that, as the number of iteration increases, the soldier and attacker's strategies $\epsilon$-like their PE counterparts, upon convergence. In the results shown in Fig. \ref{fig:subfig5:a} and Fig. \ref{fig:subfig5:b}, Bayesian updating algorithm approximately requires $210$ iterations to reach convergence. 
\begin{figure}[!t]
  \begin{center}
   \vspace{0cm}
    \includegraphics[width=8cm]{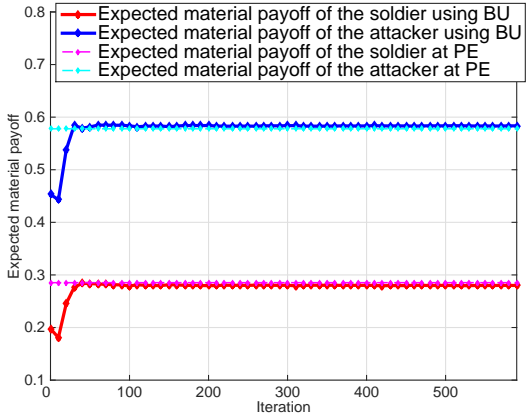}
    \caption{\label{figure6} Convergence of the BU expected material payoffs.}
  \end{center}
\end{figure}
Fig. \ref{figure6} shows the convergence of the BU expected material payoffs of the soldier and attacker, as their BU beliefs as well as their BU strategies in Fig. \ref{figure5} converge. Fig. \ref{figure5} also shows that, when the player's belief are not error-free, their strategies will not be rational, such that the player's expected material payoff could be less than their expected material payoffs at the PE. 

\begin{figure}[!t]
  \begin{center}
   \vspace{0cm}
    \includegraphics[width=8cm]{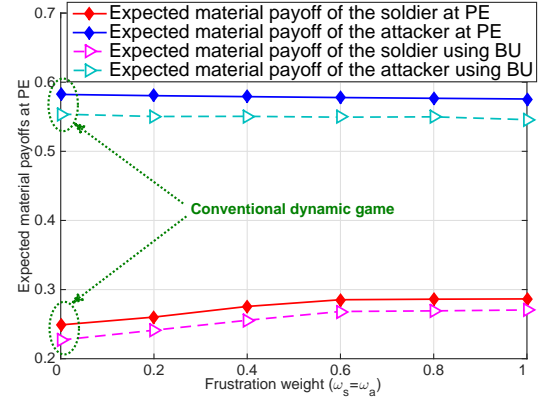}
    \vspace{-0.2cm}
    \caption{\label{figure7} The soldier's and the attacker's expected material payoffs as the frustration weights vary (BU results from 10 iterations in Algorithm 1).}
  \end{center}\vspace{-0.5cm}
\end{figure}
Fig. \ref{figure7} shows the effect of the variation in the frustration weights on the achieved expected material payoffs, of the soldier and the attacker, at PE and BU. In this regard, Fig. \ref{figure7} shows that, as the weight $\omega_a$ increases, the attacker's expected material payoff decreases, as the attacker becomes less apt to launching any attack. Fig. \ref{figure7} also shows that, as the weight $\omega_s$ increases, the soldier's  expected material payoff increased by up to $15.11\%$ as compared to the expected material payoff the soldier achieves in the conventional dynamic game (which corresponds to the PE expected material payoff of the soldier at $\omega_s=0$). This stems from the results shown in Fig. \ref{figure4} that at higher frustration weights, the attacker becomes less prone to launching attacks, which leads to an increase in the soldier's expected material payoff when the soldier's beliefs are error-free. Fig. \ref{figure7} also shows that potentially inaccurate beliefs, which occur after 10 iterations in the Bayesian updating algorithm, yield up to $9.12\%$ loss on the soldier's and the attacker's expected material payoffs.

\begin{figure}[!t]
  \begin{center}
    \includegraphics[width=8cm]{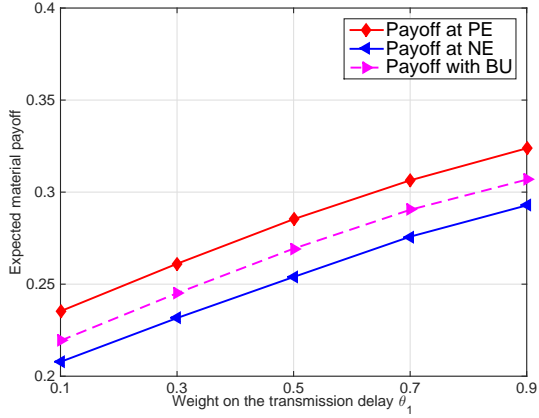}
    \caption{\label{figure8} The soldier's expected material payoff as the attacker's assigned weight, $\theta_1$, on maximizing the soldier's communication delay increases.}
  \end{center}\vspace{-0.5cm}
\end{figure}

Fig. \ref{figure8} highlights the variation in the soldier's and attacker's expected material payoffs at the equilibria as the weight, $\theta_1$, that the attacker assigns as part of its utility function to maximizing the soldier's communication delay varies. Fig. \ref{figure7} shows that, as the weight $\theta_1$ increases, the soldier's expected material payoff at the NE, PE and BU increases. Here, the BU solution is generated by running 10 iterations of the proposed Bayesian updating algorithm. In this respect, Fig. \ref{figure7} shows that, as the weight $\theta_1$ increases from $0.1$ to $0.9$, the soldier's expected material payoffs at NE and PE increase up to $37.56\%$ and $39.05\%$, respectively. This stems from the fact that the attacker's increased intention to maximize the soldier's communication delay causes an increased likelihood (at the equilibrium) that the attacker compromises the best IoBT device along the path.  This increased likelihood of attacking the IoBT device with the best channel leads to a decrease in the likelihood of attacking the remaining IoBT devices. With the error-free beliefs on the attacker, the soldier becomes more apt to connecting to the second best device along the path. As such, the soldier's expected material payoff still increases with the increased weight $\theta_1$. Meanwhile, at 10 iterations of the Bayesian updating algorithm, the resulting beliefs of each of the players would not have totally converged. Hence, at this point, the beliefs of each of the players are not completely error-free. This causes the chosen rational strategies not to completely align with the equilibrium strategies leading to a decrease of up to $9.23\%$ in the resulting expected material payoff of the soldier, as compared to its PE expected material payoff. Hence, similarly to Fig. \ref{figure7}, this result also highlights the effect that non-error-free beliefs have on the chosen strategies as well as on the outcome of the game.  

\begin{figure}[!t]
  \begin{center}
   \vspace{0cm}
    \includegraphics[width=8cm]{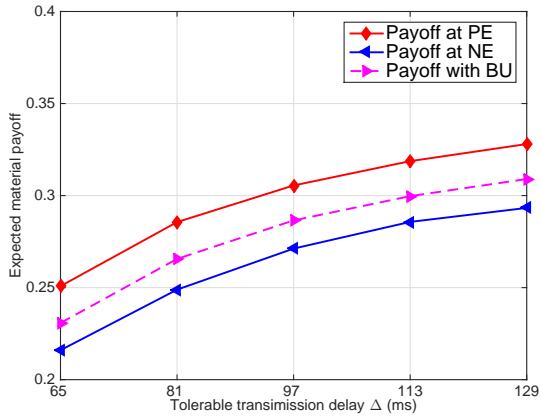}
    \caption{\label{figure9} The soldier's expected material payoff as its tolerable transmission delay varies.}
  \end{center}\vspace{-0.5cm}
\end{figure}
Fig. \ref{figure9} shows how the soldier and attacker's expected material payoffs at an equilibrium vary as the soldier's tolerable communication delay $\Delta$ increases. In this respect, Fig. \ref{figure9} shows that the soldier's expected material payoff at the NE, PE and BU increases with an increase in $\Delta$.  Fig. \ref{figure9} also shows that, as $\Delta$ increases, the soldier's expected material payoff at NE and PE increases up to $35.79\%$ and $38.11\%$, respectively, while the soldier's expected material payoff at the BU increases up to $36.88\%$. This is due to the fact that, as $\Delta$ increases from $65$ ms to $129$ ms, the difference between the normalized communication delay (i.e., $\frac{\tau}{\Delta}$) of the best channel and other channels decreases. Hence, the soldier becomes less likely, at the equilibrium, to connect to the IoBT device that has the best channel condition. As the probability that the attacker compromises the suboptimal IoBT devices decrease, the likelihood that the soldier evades the attacker's jamming attacker increases. Fig. \ref{figure9} also shows that, when the soldier and attacker choose optimal strategies based on inaccurate beliefs (e.g., after only 10 iterations of Bayesian updating), the soldier's expected material payoff will decrease by up to $8.70\%$, which is aligned with the results of Fig. \ref{figure7} and Fig. \ref{figure7}. 
\begin{figure}[!t]
  \begin{center}
   \vspace{0cm}
    \includegraphics[width=8cm]{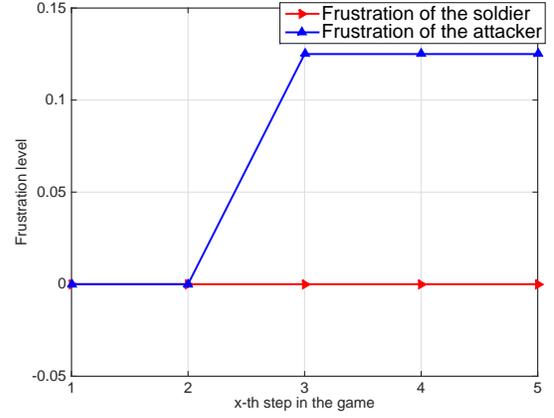}
    \caption{\label{figure10} Frustration of the soldier and the attacker at the different steps in the battlefield.}
  \end{center}\vspace{-0.3cm}
\end{figure}

Fig. \ref{figure10} shows the soldier's and attacker's frustration as the soldier progresses from one step to the other in the battlefield. At the used simulation parameters, at the third step of this mission, the attacker attacks the IoBT network, while the soldier does not connect yet to the network. After observing previous sequence of actions in the battlefield, the soldier will connect to the IoBT network at the fourth step. Fig. \ref{figure10} also shows that, as the soldier successfully connects to the IoBT network, without being jammed, the frustration level of the soldier stays at $0$. On the other hand, the frustration level of the attacker increases at step 3 since it has attacked a device to which the soldier has not communicated. Fig. \ref{figure10} shows that the frustration level of the attacker increases at step 3, and remains constant until the end of this mission.

\begin{figure}
  \centering
  \subfigure[Frustration of the soldier resulting from BU, as time elapses.]{
    \label{figure11a} 
    \includegraphics[width=8cm]{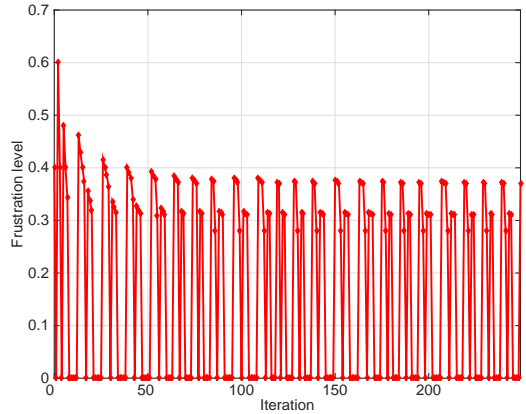}}
  \hspace{1in}
  \subfigure[Frustration of the attacker resulting from BU, as time elapses.]{
    \label{figure11b} 
    \includegraphics[width=8cm]{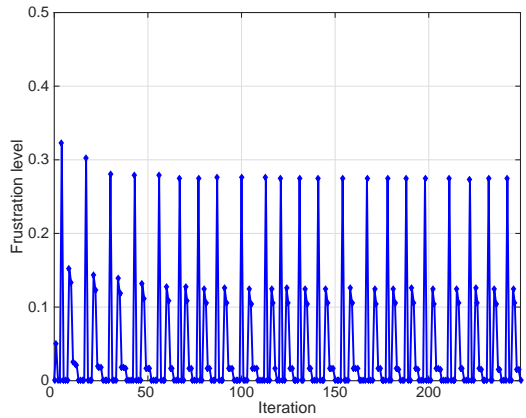}}
  \caption{Frustration of the soldier and the attacker resulting from BU, as time elapses.}
  \label{figure11} 
\end{figure}

Fig. \ref{figure11} shows the variation of the soldier's and attacker's frustration levels at different iterations of the Bayesian updating algorithm. In this regard, Fig. \ref{figure11} shows that, as time elapses, the players' frustration under each of their chosen strategies decreases, as the players' beliefs converge. This, as a result, highlights the effect of having inaccurate beliefs , not only on the expected material payoffs of each player, but also on their resulting frustration levels. 

 \vspace{0.2cm}
  \section{Conclusion}
  In this paper, we have considered an anti-jamming
problem in an IoBT network in which an adversary attempts
to interdict the connection between a soldier and IoBT devices
using jamming. We have formulated this problem as a dynamic
psychological game. Due to the reliance of the players' actions on
their beliefs, we have used the Bayesian updating to
solve this game. The psychological game enables the soldier to
determine its actions based on its estimation on the attacker's
behavior and belief. Simulation results have shown that, by
explicitly intending to frustrate its opponent, the soldier's and attacker's strategies will deviate from their strategies at a conventional, non-psychological NE. Simulation results have also shown that, using the proposed Bayesian updating algorithm, the soldier and the attacker update their beliefs toward their opponent and can reach $\epsilon$-like psychological self-confirming equilibrium (PSCE) strategies for our proposed psychological game 

  \bibliographystyle{IEEEbib}
\def\baselinestretch{0.98}
\bibliography{references}

%
%

\end{document}